\documentclass[a4paper,UKenglish,cleveref, autoref, thm-restate]{lipics-v2021}

\usepackage{tikz}

\usetikzlibrary{positioning}
\usetikzlibrary{shapes.geometric}
\usetikzlibrary{shapes.multipart}

\newcommand{\cyl}{\textnormal{cyl}}

\newcommand{\out}{\text{out}}

\newcommand{\PrefC}{\text{\upshape Pref}(C)}
\newcommand{\hpref}{\overline{h}}
\newcommand{\depth}{\text{depth}}
\newcommand{\rank}{\text{rank}}
\newcommand{\Slice}{\text{Slice}}
\newcommand{\Min}{\text{Min}}

\title{Finite-memory strategies in two-player infinite games} 

\titlerunning{Finite-memory strategies in two-player infinite games} 

\author{Patricia Bouyer}{Universit\'e Paris-Saclay, CNRS, ENS Paris-Saclay, LMF, 91190, Gif-sur-Yvette, France}{thomasset@lsv.fr}{}{}

\author{Stéphane Le Roux}{Universit\'e Paris-Saclay, CNRS, ENS Paris-Saclay, LMF, 91190, Gif-sur-Yvette, France}{leroux@lsv.fr}{}{}

\author{Nathan Thomasset}{Universit\'e Paris-Saclay, CNRS, ENS Paris-Saclay, LMF, 91190, Gif-sur-Yvette, France}{bouyer@lsv.fr}{}{}

\authorrunning{P. Bouyer, S. Le Roux and N. Thomasset} 

\Copyright{John Q. Public and Joan R. Public} 



\category{} 

\relatedversion{} 

\nolinenumbers 

\EventEditors{John Q. Open and Joan R. Access}
\EventNoEds{2}
\EventLongTitle{Computer Science Logic 2022}
\EventShortTitle{CSL'22}
\EventAcronym{CSL}
\EventYear{2022}
\EventDate{February 14--19, 2022}
\EventLocation{Göttingen, Germany}
\EventLogo{}
\SeriesVolume{42}
\ArticleNo{80}

\begin{document}

\maketitle

\begin{abstract}
We study infinite two-player win/lose games $(A,B,W)$ where $A,B$ are
finite and $W \subseteq (A \times B)^\omega$. At each round Player 1
and Player 2 concurrently choose one action in $A$ and $B$, respectively. Player 1 wins iff the generated sequence is in $W$. Each history $h \in (A \times B)^*$ induces a game $(A,B,W_h)$ with $W_h := \{\rho \in (A \times B)^\omega \mid h \rho \in W\}$. We show the following: if $W$ is in $\Delta^0_2$ (for the usual topology), if the inclusion relation induces a well partial order on the $W_h$'s, and if Player 1 has a winning strategy, then she has a finite-memory winning strategy. Our proof relies on inductive descriptions of set complexity, such as the Hausdorff difference hierarchy of the open sets.

Examples in $\Sigma^0_2$ and $\Pi^0_2$ show some tightness of our result. Our result can be translated to games on finite graphs: e.g. finite-memory determinacy of multi-energy games is a direct corollary, whereas it does not follow from recent general results on finite memory strategies.
\end{abstract}

\section{Introduction}

Two-player win/lose games have been a useful tool in various areas of logic and computer science. The two-player win/lose games in this article consist of infinitely many rounds. At each round $i$, Player 1 and Player 2 concurrently choose one \emph{action} each, i.e. $a_i$ and $b_i$ in there respective sets $A$ and $B$. Player 1 wins and Player 2 loses if the play $(a_0,b_0)(a_1,b_1)\dots$ belongs to a fixed $W\subseteq (A \times B)^\omega$. Otherwise Player 2 wins and Player 1 loses. We call $W$ the \emph{winning set} of Player 1, or the winning condition for Player 1. A \emph{strategy} is a map that tells a player how to play after any finite \emph{history} of actions played: a Player 1 (resp. 2) strategy is a map from $(A \times B)^*$ to $A$ (resp. $B$). A strategy is \emph{finite-memory} (FM) if the map can be implemented by a finite-state machine. Also, each history $h \in (A \times B)^*$ induces a game starting at $h$ and taking the past into account, i.e. with winning set $W_h := \{\rho \in (A \times B)^\omega \mid h \rho \in W\}$. 

For now we state a slightly weaker version of our main result: if $A$ and $B$ are finite, if the $(W_h)_{h \in (A \times B)^*}$ constitute a well partial order (wpo) for the inclusion, if $W\in \Delta^0_2$, i.e. in the usual cylinder topology $W$ is a countable union of closed sets and a countable intersection of open sets, and if Player 1 has a winning strategy, then she has a finite-memory winning strategy. Of course, our result also applies to the turn-based version of such games.

{\bf On the proof of the main result:} The proof of our main result
relies on descriptive set theory. The Hausdorff-Kuratowski theorem
(see, e.g., \cite{Kechris95}) states that each set in $\Delta^0_2$ can
be expressed as an ordinal difference of open sets, and conversely. In
general, this implies that properties of sets in $\Delta^0_2$ may be
proved by induction over the countable ordinals. Accordingly, we prove
our main result by induction on $W$, but the inductive step suggested
by the Hausdorff-Kuratowski theorem does not suit us
completely. Instead, we mix it with a folklore alternative way of describing
$\Delta^0_2$ by induction. After this mix, our base case consists of the open sets,
the first inductive step consists of \emph{union with a closed set}, and the
second inductive step of \emph{open union}.

To prove our result, the base case, where $W$ is open, amounts to reachability games, and the wpo assumption is not needed. The case where $W$ is closed is easy, and it includes the multi-energy games (where Player 1 keeps all energy levels positive). Just above these, the case 
where $W$ is the union of an open set and a closed set is harder to prove. It includes disjunctions of a reachability condition and a multi-energy condition. We will present this harder case in details because it shows part of the complexity of the full result.

%
%
%

{\bf A fourth representation of $\Delta^0_2$:} Above, we mentioned two
hierarchies that describe $\Delta^0_2$. In addition, this paper
(re-)prove the folklore result that $\Pi^0_2$ corresponds to Büchi
winning conditions and $\Sigma^0_2$ to co-Büchi. If labeling each
history with $0$ or $1$, the Büchi (co-Büchi) condition requires that
infinitely (only finitely) many $1$'s be seen on a branch/play. The
$\Delta^0_2$ sets are therefore the sets that can be expressed both by
Büchi and co-Büchi conditions. This is possible exactly if on every


{\bf Tightness of the result:} The collection of the countable unions of closed sets is called $\Sigma^0_2$, and the collection of their complements, i.e. the countable intersections of open sets, is called $\Pi^0_2$. So $\Delta^0_2 = \Sigma^0_2 \cap \Pi^0_2$. In this article we provide one example of a winning set $W$ in $\Sigma^0_2$ and one example in $\Pi^0_2$ that satisfy the wpo assumption but not the FM-strategy sufficiency. Hence tightness.

Note that without the wpo assumption, even Turing-computable strategies may not suffice to win for closed winning sets: take a non-computable binary sequence $\rho$ and a game where Player 1 wins iff she plays $\rho$. She has a winning strategy, but no computable ones.

\textbf{Connections with graph games:} Our game $\langle A,B,W \rangle$ can be seen as a one-state concurrent graph game where the winning condition is defined via the actions rather than the visited states. Winning strategies and FM strategies coincide in both models.

Alternatively, we can unfold any concurrent graph game into an infinite tree game $\langle A,B,W \rangle$ whose nodes are the histories of pairs of actions. Winning strategies coincide in both models. Moreover, an FM strategy in the tree game is, up to isomorphism, also an FM strategy in the graph game. The converse may not hold since, informally, the player may observe the current state only in the graph model. Nevertheless for finite graphs, the observation of the state can be simulated by an additional finite memory, i.e. the graph itself. To sum up, any FM result in our tree games can be translated into an FM result in graph games. Almost conversely, FM results in finite-graph games can be obtained from our tree games, possibly with non-optimal memory.

%

\textbf{Related works and applications:} The two articles \cite{LPR18} and \cite{BLORV20} provide abstract criteria to show finite-memory determinacy in finite-graph games: \cite{LPR18} by Boolean combination of complex FM winning conditions with simple winning conditions defined via regular languages; \cite{BLORV20} by characterizing, for a fixed memory, the winning conditions that yield, in all finite-graph games, FM determinacy via this fixed memory. The FM determinacy of multi-energy games is a corollary of neither, but as mentioned above, it is a direct corollary of our result.

More generally on a finite-graph game, consider the conjunction or disjunction of a multi-energy winning condition and a Boolean combination of reachability conditions. This is in $\Delta^0_2$ (actually low at some finite level of the hierarchy), this induces a wpo, so if Player 1 has a winning strategy she has a FM one.

However, the FM determinacy of Büchi games in finite-graph games is not a corollary of our result, because the Büchi conditions may not be in $\Delta^0_2$. We mention three things about this. First, this determinacy does not contradict our tightness results: in $\Pi^0_2$ or $\Sigma^0_2$, FM determinacy holds when the corresponding labeling is regular. Second, in future work we plan to seek a general theorem having both this determinacy and our main result as special cases. Third, in finite-graph games, many winning conditions that yield memoryless or FM determinacy can be simulated by finite games, and therefore clopen winning conditions, i.e. $\Delta^0_1$ instead of our more general $\Delta^0_2$. For instance, see \cite{AR17}, \cite{McNaughton00}, \cite{FZ12}. This suggests that our work could be used to prove more FM sufficiency results by reduction of finite-graph games to tree games with wpo winning condition in $\Delta^0_2$.

{\bf Structure of the article:} Section~\ref{sec:def} defines our games and finite-memory strategies; Section~\ref{sec:dst} presents the related descriptive set theory; Section~\ref{sec:ex-fms} presents our main result; Section~\ref{sec:tightness} discusses tightness of our main result; Section~\ref{sec:conclusion} mentions possible future work.

\section{Setting and definitions}\label{sec:def}

We study two-player games, which consist in a tuple $(A,B,W)$: $A$ is the \textit{action set} for Player $1$, $B$ is the action set for Player $2$ and $W\subseteq (A\times B)^\omega$ is the \textit{winning   set}. Here we only consider finitely branching games, where $A$ and $B$ are both finite. Such a game is played in the following way: at each round, each player chooses an action from their respective action set in a concurrent way, thus producing a pair of actions in $A\times B$. The game then continues for infinitely many rounds, generating a \textit{play} which consists in an infinite word in $(A\times B)^\omega$. Player $1$ then wins if the generated play belongs to the winning set $W$, while Player $2$ wins if it does not. In the following we will focus on Player $1$.

To describe the Players' behavior in such a game, we use the concepts of \textit{histories} and \textit{strategies}. A history is a finite word in $(A\times B)^*$ and represents the state of the game after finitely many rounds. We call $\mathcal{H}$ the set of histories. A strategy $s: \mathcal{H} \to A$ for Player $1$  is a function that maps histories to actions and represents a behavior for Player $1$: in history $h$, she will play action $s(h)$. Given a strategy $s$, a history $h$ and a word $\beta\in B^\omega$, we call $\out(h,s,\beta)$ the only play where both players first play $h$, then Player $1$ plays according to $s$ and Player $2$ plays the actions of $\beta$ in order. This play is defined inductively as follows:
\begin{itemize}
	\item for $k\leq |h|$, $\out(h,s,\beta)_{< k} = h_{< k}$;
	\item for $k > |h|$, $\out(h,s,\beta)_{< k} = \out(h,s,\beta)_{< k-1}(s(\out(h,s,\beta)_{< k-1}),\beta_{k-|h|-1})$.
\end{itemize}
We say that  a play $\rho$ is \textit{compatible} with a given strategy $s$ if there exists $\beta\in B^\omega$ such that $\rho = \out(\varepsilon,s,\beta)$, where $\varepsilon$ is the empty history. Similarly, a history is compatible with $s$ if it is the finite prefix of a compatible play. We say that a strategy is \textit{winning} if all the plays compatible with it belong to $W$: if Player $1$ plays according to such a strategy, she is guaranteed to win. We can extend this concept to say that a strategy $s$ is winning \textit{from a history} $h$ when for all $\beta\in B^\omega$ we have $\out(h,s,\beta)\in W$ (if in history $h$ Player $1$ starts playing according to $s$ then she will win). We call a \textit{winning history} a history from which there exists a winning strategy.

We call a \textit{tree} any subset of $(A\times B)^*$ which is closed by prefix, and a branch any (finite or infinite) sequence of elements $e_0 = \varepsilon \sqsubset e_1 \sqsubset e_2 \sqsubset...$ of a tree. In particular, all histories compatible with a given strategy form a tree, which we call the \textit{strategic tree} induced by the strategy. We makes extensive use of Kőnig's lemma \cite{konig}, which states that if a tree has no infinite branch then it is a finite tree. Specifically, we often use the derived result that if some family in a tree intersects all infinte branches of the tree then it has a finite subset that also does.

Given a history $h \in \Gamma$, we say that an action $a \in A$ is non-losing for $h$ if for any action $b\in B$,  $h(a,b)$ is a winning history. Among all histories, we are particularly interested in the set of histories along which Player $1$ has only played non-losing actions: we call $\Gamma$ this particular set. Notice that Player $1$ has a winning strategy from any history $h\in\Gamma$, but that playing only non-losing for Player $1$ might be a losing strategy.

A history $h \in \mathcal{H}$ \textit{induces} a winning set $W_h$ defined as $W_h = \{\rho\in (A\times B)^\omega \mid h\rho\in W\}$. The set $W_h$ contains all the infinite continuations $\rho$ such that $h\rho$ is a winning play.

Recall that we introduced strategies for Player $1$ as functions mapping histories to actions in $A$. Among these strategies, we are particularly interested in those that can be described as finite machines: we call them \textit{finite-memory} strategies. Let us introduce first the concept of \textit{finite-memory decision machines}. A finite-memory decision machine is a tuple $(M,\sigma,\mu,m_0)$ such that:
\begin{itemize}
	\item $M$ is a finite set (the memory);
	\item $\sigma: M \to A$ is the decision function;
	\item $\mu: M\times(A\times B)\to M$ is the memory update function;
	\item $m_0\in M$ is the initial memory state.
\end{itemize}

Given a finite-memory decision machine $(M,\sigma,\mu,m_0)$, we extend $\mu$ by defining $\mu(m,h)$ for $h \in \mathcal{H}$ in the following inductive way:
\begin{itemize}
	\item $\mu(m,\varepsilon) = m$
	\item for all $h$ in $\mathcal{H}$ and $(a,b)\in A\times B$, $\mu(m,h(a,b))=\mu(\mu(m,h),(a,b))$.
\end{itemize}
For readability's sake, in case $m=m_0$ and the context is clear, we write $\mu(h)$ for $\mu(m,h)$.

A finite-memory decision machine $(M,\sigma,\mu,m_0)$ induces a strategy $s$ for Player $1$ defined for all $h \in \mathcal{H}$ as $s(h) = \sigma(\mu(m_0,h))$. We say that a strategy $s$ is a finite-memory strategy if it is induced by some finite-memory decision machine, and abusively write $s = (M,\sigma,\mu,m_0)$ (identifying the finite-memory decision machine with the strategy it induces) when it is the case.

We often describe a finite-memory decision machine by only defining $\mu$ for the action pairs that are compatible with $\sigma$ (i.e. for $m\in M$ we only define $\mu(m,(a,b))$ when $a = \sigma(m)$). Such a partial machine can be easily extended to a complete one, and contains all the relevant information to decide on the winning aspect of the strategy (or rather, strategies, as many different extensions are possible) it induces as it describes all plays compatible with itself.

\section{Descriptive set theory}\label{sec:dst}

\subsection{Open sets and the Borel hierarchy}

Given a set $C$ and a finite word $w\in C^*$, we call \textit{cylinder} of $w$ the set $\cyl(w) = \{w\rho \mid \rho\in C^\omega\}$. This set contains all the infinite words that start with $w$. In concordance with the usual cylinder topology on $C^\omega$, the cylinders serve as the basis for the open sets, in the sense that we define as an \textit{open set} any set that can be written as an arbitrary union of cylinders. We say that a family of words $\mathcal{F}$ is a \textit{generating family} for an open set $O$ if we have $O = \cup_{f\in \mathcal{F}} \cyl(f)$, that is, $O$ is the set of all plays that have at least one finite prefix in $\mathcal{F}$.

These open sets allow to define a Borel algebra on $C^\omega$ as the smallest $\sigma$-algebra that contains all open sets. More precisely, the Borel algebra is the smallest collection of sets that contains the open sets and is closed under both countable union and complement (for more information about Borel sets, see \cite{Kechris95}). This collection of sets can be organized into what is called the Borel hierarchy, which is defined for countable ordinals in the following way:
\begin{itemize}
	\item $\Sigma_1^0$ is the collection of all the open sets;
	\item for all countable ordinals $\theta$, $\Pi_\theta^0$ is the collection of sets whose complements are in $\Sigma_\theta^0$;
	\item for all countable ordinals $\theta$, $\Sigma_\theta^0$ is the collection of sets that can be defined as a countable union of sets belonging to lower levels of the hierarchy;
	\item finally, for all countable ordinals $\theta$, $\Delta_\theta^0$ is the collection of sets that are in both $\Sigma_\theta^0$ and $\Pi_\theta^0$.
\end{itemize}
To illustrate, let us detail the lowest levels of the hierarchy:
\begin{itemize}
	\item as per the definition, $\Sigma_1^0$ is the collection of all the open sets;
	\item the sets in $\Pi_1^0$ are the sets whose complement is an open set, we call them the \textit{closed} sets;
	\item $\Sigma_2^0$ contains the sets which can be written as a countable union of closed sets;
	\item $\Pi_2^0$ contains the sets which complement can be written as a countable union of closed sets: by properties of the complement, these are the sets that can be written as a countable intersection of open sets;
	\item finally, $\Delta_2^0$ is the collection of sets that can be written both as a countable union of closed sets and as a countable intersection of open sets.
\end{itemize}
In the following we will focus on the collection of sets $\Delta_2^0$.

\subsection{The Hausdorff difference hierarchy}

The Hausdorff difference hierarchy (see for instance \cite{Kechris95}) provides us with a way of defining inductively all the sets in $\Delta_2^0$. Formally, given an ordinal $\theta$ and an increasing sequence of open sets $(O_\eta)_{\eta<\theta}$, the set $D_\theta((O_\eta)_{\eta<\theta})$ is defined by:
\[
\rho\in D_\theta((O_\eta)_{\eta<\theta}) \quad \Leftrightarrow \quad \begin{array}{@{}l} \rho\in\cup_{\eta<\theta}O_\eta\ \text{and the least}\ \eta\ \text{such that}\ \rho\in O_\eta\\ \text{has parity opposite to that of}\ \theta. \end{array}
\] 

For any ordinal $\theta$, we call $\mathcal{D}_\theta$ the collection of sets $S$ such that there exists an increasing family of open sets $(O_\eta)_{\eta<\theta}$ such that $S = D_\theta((O_\eta)_{\eta<\theta})$. To illustrate, $\mathcal{D}_1$ is the collection of all the open sets, $\mathcal{D}_2$ is the collections of the sets that can be written as $O_1\setminus O_0$ where $O_1$ and $O_0$ are two open sets (and hence contains the closed sets), $\mathcal{D}_3$ is the collection of the sets that can be written as $O_2\setminus (O_1\setminus O_0)$ where $O_2$, $O_1$ and $O_0$ are three open sets, etc.

The Hausdorff-Kuratowski theorem \cite{Kechris95} then states that a set $S$ belongs to $\Delta_2^0$ if and  only if there exists an ordinal $\theta$ such that $S\in\mathcal{D}_\theta$.

%
%
%
%

\subsection{The fine Hausdorff hierarchy}

In the spirit of the Hausdorff difference hierarchy, we propose another inductive way of defining the sets in $\Delta_2^0$. This other hierarchy was already introduced in \cite{SLR15}, and might have appeared earlier in the literature but to our knowledge has never been studied in similar depth. First we introduce the concept of \textit{open union}: we say that the union of a family of sets $(S_i)_{i\in I}$ is an open union if there exists a family of disjoint open sets $(O_i)_{i\in I}$ such that for all $i\in I$ have $S_i\subseteq O_i$. We denote such a union by $\Cup_{i\in I} S_i$.

We then define inductively collections of sets $\Lambda_\theta$ and $K_\theta$, with $\theta$ a positive ordinal, in the following way:
\begin{itemize}
	\item a set $S$ is in $\Lambda_1$ if and only if it is an open set;
	\item a set $S$ is in $K_\theta$ if and only if its complement is in $\Lambda_\theta$;
	\item a set $S$ is in $\Lambda_\theta$ with $\theta > 1$ if and only if there exists a family of sets $(S_i)_{i\in I}$ such that for each $i$ there exists $\eta_i < \theta$ such that $S_i\in\Lambda_{\eta_i}\cup K_{\eta_i}$ and we have $S = \Cup_{i\in I} S_i$.
\end{itemize}

This definition is akin to the definition of the Borel hierarchy, with the exception that the union operation is replaced with an open union.
We then prove the following theorem, which shows our hierarchy is a refinement of the Hausdorff difference hierarchy (and hence justifies its name):

\begin{theorem}\label{THM:fineHausdorffIsHausdorffCore}
	For all ordinals $\theta$, we have $\mathcal{D}_\theta=\Lambda_\theta$.
\end{theorem}

This theorem is naturally proven by induction on $\theta$ and requires intermediate results which help understand the nature of the two hierarchies. In particular, we have: (i) for all ordinals $\theta$ the collection $\mathcal{D}_\theta$ is closed under open union, (ii) for all ordinals $\theta$ the collection $\Lambda_\theta$ is closed by intersection with an open set, (iii) for all ordinals $\theta$ the two collections $\Lambda_\theta$ and $K_\theta$ are closed under intersection with a cylinder and
(iv) for all limit ordinals $\theta$ the collection $\mathcal{D}_\theta$ is the collection of sets that can be written as  the open unions of sets in $\cup_{\eta < \theta}\mathcal{D}_\eta$. One observation which proves pivotal for proving our main result on the existence of finite-memory strategies is that for all ordinals $\theta$, all sets in $K_\theta$ can be written as the union of a closed set and a set that belongs to $\Lambda_\theta$ (if $\theta$ is a successor odinal, we can be even more precise as $K_\theta$ is the collection of all sets that can be written as the union of a closed set and a set in $\Lambda_{\theta-1}$). All the details surrounding these two views and the proof of theorem \ref{THM:fineHausdorffIsHausdorffCore} can be found in appendix \ref{AP:descriptiveSetTheory}.

%
%
%
%

\subsection{The 0-1 eventually constant labelling}

A third possible view of sets in $\Delta_2^0$ is given via eventually constant labelling functions. We say that a labelling function $l: C^*\to \{0,1\}$ is eventually constant if for all infinite words $\rho$ in $C^\omega$, the sequence $(l(w_{<n}))_{n\in\mathcal{N}}$ is eventually constant, which means that there exists a finite $k\in \mathcal{N}$ and $i\in \{0,1\}$ such that for all $n\geq k$ we have $l(w_{< n}) = i$. We then call $1_l$ the set of infinite words $w\in C$ such that the set $\{n\mid l(w_{< n}) = 1\}$ is infinite. As we will see, the sets $S$ that belong to $\Delta_2^0$ are the sets such that there exists an eventually constant labelling function $l$ such that $S = 1_l$.

\subsection{Equivalence of representations}

\subsubsection{Representations of sets in $\Delta_2^0$}

As expressed by the following theorem, the Hausdorff difference hierarchy, fine Hausdorff hierarchy and eventually constant labelling functions actually define the same sets, which are exactly the sets that belong to $\Delta_2^0$.

\begin{theorem} \label{THM:delta02Representations}
	Given a subset $S$ of $C^\omega$, the following propositions are equivalent:
	\begin{itemize}
		\item[$(1)$] $S\in\Delta_2^0$;
		\item[$(2)$] $S$ belongs to the Hausdorff difference hierarchy;
		\item[$(3)$] $S$ belongs to the fine Hausdorff hierarchy;
		\item[$(4)$] there exists an eventually constant labelling function $l$ such that $S = 1_l$.
	\end{itemize}
\end{theorem}

The detailed proof can be found in Appendix~\ref{AP:equivalenceOfRepresentations}, but we give some elements here: the Hausdorff-Kuratowski theorem \cite{Kechris95} shows that $(1)\Leftrightarrow(2)$, and Theorem \ref{THM:fineHausdorffIsHausdorffCore} shows that $(2)\Leftrightarrow(3)$. We prove that $(3)\Rightarrow (4)$ by showing that the collection of sets of the form $1_l$ where $l$ is an eventually constant labelling function is closed under both open union and complement, and finally prove that $(4)\Rightarrow (1)$ by showing that all sets of the form $1_l$ where $l$ is an eventually constant labelling function can be expressed both as countable intersection of open sets and countable union of closed sets.

\subsubsection{Correspondence between Büchi/co-Büchi conditions and $\Pi_0^2$/$\Sigma_0 ^2$}

Two much studied types of winning condition in computer science are the \textit{Büchi} and \textit{co-Büchi} conditions. Such winning conditions are given by a coloring function $c$ that provides a color (elements in $\{0,1\}$) for every history. In the case of a Büchi condition, a play is then winning if infinitely many of its prefixes are associated with the color $1$ while in the case of a co-Büchi condition it is winning if finitely many of its prefixes are associated with the color $1$ (Büchi and co-Büchi conditions are thus the complement of each other). As stated by the following lemma, whose proof can be found in Appendix~\ref{app:Buchi}, Büchi conditions actually describe the sets in $\Pi_2^0$:

\begin{lemma}\label{LM:BuchiIsPi02}
	A subset $S$ of $C^\omega$ belongs to $\Pi_2^0$ if and only if it can be expressed as a Büchi condition.
\end{lemma}

A trivial corollary is that co-Büchi conditions describe the sets in $\Sigma_2^0$:

\begin{corollary}
	$W$ belongs to $\Sigma_2^0$ if and only if it can be expressed as a co-Büchi condition.
\end{corollary}

\section{On the existence of finite-memory winning strategies when the winning set belongs to the Hausdorff difference hierarchy}\label{sec:ex-fms}

Our aim is to exhibit conditions on $W$ that ensure Player $1$ has a finite-memory winning strategy when some winning strategy exists.

Consider a game $(A,B,W)$ where the winning set $W$ belongs to $\Delta_2^0$. We introduce a new hypothesis on the induced winning sets of this game: the set inclusion relation, denoted by $\subseteq$, induces a well partial order (wpo) on the winning sets induced by the histories in $\Gamma$. That is, for any sequence $(h_n)_{n\in\mathbb{N}}$ of histories in $\Gamma$, there exists $k < l$ such that $W_{h_k}\subseteq W_{h_l}$. A known property of well partial orders which we will use is that any set $S \subseteq \Gamma$ contains a finite subset $M$ such that for all $h\in S$ there exists $m\in M$ such that $W_m\subseteq W_h$. The set of winning sets induced by the histories in $M$ effectively functions as a finite set of under-approximations for the winning sets induced by the histories in $S$.

Such hypotheses might seem exotic and restrictive, but are effectively satisfied for well-studied classes of games, such as energy games or multi-energy games played on graphs (see for instance \cite{energyGames}), or games with a winning condition expressed as a boolean combination of reachability/safety conditions.

Under these specific conditions, we prove that Player $1$ always has a finite-memory winning strategy when she has a winning strategy:
\begin{theorem}\label{thebigtheorem}
	Assume that $W$ belongs to $\Delta_2^0$ and $\subseteq$ induces a well partial order on $\{W_h \mid h\in\Gamma\}$. If Player~$1$ has a winning strategy from $\varepsilon$, then she also has a finite-memory one.
\end{theorem}

Given a set $S$ in the Hausdorff difference hierarchy, the rank of $S$ is the least ordinal $\theta$ such that $S\in\mathcal{D}_\theta$. 
We prove Theorem~\ref{thebigtheorem} by a transfinite induction on the rank of $W$.

\smallskip First notice that the inclusion of induced winning sets has
the nice property of being preserved by the addition of a suffix,
which is formally expressed by the following lemma:

\begin{lemma}\label{orderTransitive}
	If $W_h\subseteq W_{h'}$ then for all $(a,b)\in A\times B$ we have $W_{h(a,b)}\subseteq W_{h'(a,b)}$.
\end{lemma}

\begin{corollary}\label{COR:non-losingIncluded}
	If $W_h\subseteq W_{h'}$ then all non-losing actions of $h$ are also non-losing for $h'$.
\end{corollary}

\subsection{Proof for open sets}

We begin by the case where $W$ is an open set ($W$ has rank $1$), generated by a set $\mathcal{F}$  of histories. 

\begin{lemma}\label{finitememOpen}
	If $W$ is an open set and $\varepsilon \in \Gamma$, then Player $1$ has a finite-memory winning strategy from $\varepsilon$.
\end{lemma}

\begin{proof}
	Suppose that there exists a winning strategy $s$ from $\varepsilon$. Then consider the strategic tree $T$ induced by $s$, and consider the tree $T' = T\setminus\{h\in\mathcal{H}\mid\exists f\in \mathcal{F}, f\sqsubset h\}$, where $\sqsubset$ is the strict prefix relation. Since $s$ is winning, there is no infinite branch in $T'$. By Kőnig's lemma, this means that $T'$ is finite and by definition all maximal elements (with regards to $\sqsubseteq$, the prefix relation) of $T'$ belong to $\mathcal{F}$. $T'$ can then serve as the memory of a finite-memory winning strategy $s_f = (T', \sigma, \mu, m_0)$ defined by:
	\begin{itemize}
		\item for $t\in T'\setminus\mathcal{F}$, $\sigma(t) = s(t)$;
		\item for $t\in T'\cap\mathcal{F}$, we set $\sigma(t)$ as any action $a\in A$;
		\item for $t\in T'$ and $b\in B$, $\mu(t,(\sigma(t),b)) = t(\sigma(t),b)$ if $t\notin \mathcal{F}$ and  $\mu(t,(\sigma(t),b)) = t$ if $t\in\mathcal{F}$;
		\item $m_0 = \varepsilon$.
	\end{itemize}
	The strategy $s_f$ works by simply following $s$ alongside the branches of $T'$ until it reaches a history in $\mathcal{F}$, and is thus winning.
\end{proof}

\subsection{Proof for closed sets}

We now focus on the study of the case where the winning set $W$ is a closed set. 
In that case, the plays $\rho$ such that $\rho\in W$ are precisely the plays for which all finite prefixes $h$ are such that $W_h\neq \emptyset$. As a consequence, any strategy playing non-losing actions for Player $1$ is a winning strategy: such a strategy only generates histories $h$ in $\Gamma$, and in particular $W_h \ne \emptyset$.

Furthermore, if $\subseteq$ induces a well partial order on the partial winning sets, then:
\begin{enumerate}
	\item [$(*)$] there exists a finite subset $M$ of $\Gamma$ such that for all $h\in \Gamma$ there exists $m\in M$ such that $W_m\subseteq W_h$;
	\item [$(**)$] any play $\rho$ has two finite prefixes $\rho_0$ and $\rho_1$ such that $\rho_0 \sqsubset \rho_1$ and $W_{\rho_0}\subseteq W_{\rho_1}$.
\end{enumerate}
These two observations are the basis for two different approaches to prove the next lemma.

\begin{lemma}\label{finitememClosed}
	If $W$ is an open set, if  $\subseteq$ induces a well partial order on $(W_h)_{h \in \Gamma}$  and if $\varepsilon \in \Gamma$, then Player $1$ has a finite-memory winning strategy from $\varepsilon$.
\end{lemma}

\begin{proof}
Consider indeed a game $(A,B,W)$ where $W$ is a closed set, $\subseteq$ induces a well partial order
on the partial winning sets associated with the winning histories  and such that $\varepsilon$ is a winning history. We consider a strategy $s$ that is a winning strategy for Player $1$.

The first approach, derived from observation $(*)$, consists in building a finite-memory strategy $(M,\sigma,\mu,m_0)$ with memory set $M$ in the following way:
\begin{itemize}
	\item for $m\in M$, we let $\sigma(m)$ be any non-losing action from $m$,
	\item for $m\in M$ and $b\in B$, since $\sigma(m)$ is non-losing we know that $m(\sigma(m),b)\in\Gamma$, which means that there exists $m' \in M$ such that $W_{m'} \subseteq W_{m(\sigma(m),b)}$; we then let $\mu(m,(\sigma(m),b)) = m'$;
	\item finally $m_0 \in M$ is chosen such that we have $W_{m_0}\subseteq W_\varepsilon$.
\end{itemize}
Informally, we have as our memory the set $M$ which contains under-approximations for all winning sets induced by the histories of $\Gamma$. We use the transition function to maintain an under-approximation of the "real" induced winning set associated to the current history, and play according to this under-approximation. By Corollary \ref{COR:non-losingIncluded}, this ensures that we always play a non-losing action, which is enough to guarantee the win because $W$ is a closed set. 

\smallskip
The second approach is derived from observation $(**)$. Consider the winning strategy $s$ and its associated strategic tree $T_s$. Along every infinite branch $\rho$ of $T_s$, there exist two histories $h,h'$ such that $h\sqsubset h'$ and $W_h\subseteq W_{h'}$. Consider then the tree $T_s^f$ obtained by pruning $T_s$ along these histories: $T_s^f = \{h'\in T_s \mid \forall h\in T_s, h\sqsubset h'\Rightarrow W_h\nsubseteq W_{h'}\}$. By Kőnig's lemma, $T_s^f$ is a finite tree. We call $\mathcal{P}$ the set $\{h\in T_s \mid h\notin T_s^f \text{ and } h'\sqsubset h \Rightarrow h'\in T_s^f\}$ of the minimal elements (with regards to the prefix relation) of $T_s$ that do not belong to $T_s^f$. We then build a finite-memory strategy $(M',\sigma,\mu,m_0)$ in the following way:
\begin{itemize}
	\item $M'= T_s^f$
	\item for $m\in M'$, we let $\sigma(m) = s(m)$,
	\item for $m\in M'$ and $(a,b)\in A\times B$ such that $a = \sigma(m) = s(m)$,
	\begin{itemize}
		\item if $m(a,b) \in T_s^f$ then we let $\mu(m,(a,b)) = m(a,b)$,
		\item else by construction we have $m(a,b)\in \mathcal{P}$ and there exists $m'\in T_s^f$ such that $m' \sqsubset m(a,b)$ and $W_m\subseteq W_{m(a,b)}$; we then let $\mu(m,(a,b)) = m'$,
	\end{itemize}
	\item finally $m_0 = \varepsilon$.
\end{itemize}
Informally, this approach consists in playing according to $s$ until we reach a history whose induced winning set is bigger than one we already met. We then forget the current history and continue playing as if we were in the history with the smaller induced winning set. This second approach also ensures that the memory consists of an under-approximation of the "real" induced winning set, and hence by Corollary \ref{COR:non-losingIncluded} it guarantees that the resulting strategy is non-losing, and thus winning since $W$ is a closed set.
\end{proof}
 
\subsection{Limitations to the above approaches}

Until now, we have studied the lowest levels of the Hausdorff difference hierarchy, focusing on the cases where the winning sets belongs to $\Lambda_1$, the open sets, and $K_1$, the closed sets. We will explain later how to handle the case for $\Lambda_2$ and for now turn our attention to $K_2$, as it proves pivotal to the understanding of our method. 

The sets in $K_2$ are the sets that can be written as the union of a closed set and an open set. Informally, this means that Player $1$ can win in two different ways, by ensuring that either the generated play lies in the closed set or they reach a history which belongs to the generating family of the open set.

The first condition is akin to a safety objective (Player $1$ manages to never go out of a certain region) while the second condition is akin to a reachability objective (Player $1$ meets a certain given condition at a finite time and it suffices to ensure the win). As shown by the following example, the two simple approaches we detailed previously for closed sets do not suffice here: 

\begin{figure}
	\begin{center}
		\begin{tikzpicture}[-latex, auto, node distance = 1.5 cm and 1.5 cm, on grid, semithick, trianglenode/.style={draw, triangle, color = red}, history/.style={draw, circle, fill, scale = .7}, every text node part/.style={align=center}]
		
		\node[history] (epsilon) {};
		
		\node[history] [below left = 1.5 and .75 of epsilon] (01) {};
		\node [below = .35 of 01] (01winningSet) {\small \textcolor{red}{$(\{0\}\times B)^\omega$}};
		\node[history] [left = 1.5 of 01] (00) {};
		\node[history] [below right = 1.5 and .75 of epsilon] (10) {};
		\node [below = .35 of 10] (10winningSet) {\small \textcolor{red}{$\emptyset$}};
		\node[history] [right = 1.5 of 10] (11) {};
		\node [below = .35 of 11] (11winningSet) {\small \textcolor{red}{$\emptyset$}};
		
		\node[history] [below left = 1.5 and .75 of 00] (second01) {};
		\node [below = .35 of second01] (second01winningSet) {\small \textcolor{red}{$(\{0\}\times B)^\omega$}};
		\node[history] [left = 1.5 of second01] (second00) {};
		\node[history] [below right = 1.5 and .75 of 00] (second10) {};
		\node [below = .35 of second10] (second10winningSet) {\small \textcolor{red}{$\emptyset$}};
		\node[history] [right = 1.5 of second10] (second11) {};
		\node [below = .35 of second11] (second11winningSet) {\small \textcolor{red}{$\emptyset$}};
		
		\node[history] [below left = 1.5 and .75 of second00] (third01) {};
		\node [below = .35 of third01] (third01winningSet) {\small \textcolor{red}{$(A\times B)^\omega$}};
		\node[history] [left = 1.5 of third01] (third00) {};
		\node[history] [below right = 1.5 and .75 of second00] (third10) {};
		\node [below = .35 of third10] (third10winningSet) {\small \textcolor{red}{$(A\times B)^\omega$}};
		\node[history] [right = 1.5 of third10] (third11) {};
		\node [below = .35 of third11] (third11winningSet) {\small \textcolor{red}{$(A\times B)^\omega$}};
		
		\node[history] [below left = 1.5 and .75 of third00] (fourth01) {};
		\node [below = .35 of fourth01] (fourth01winningSet) {\small \textcolor{red}{$(A\times B)^\omega$}};
		\node[history] [left = 1.5 of fourth01] (fourth00) {};
		\node[history] [below right = 1.5 and .75 of third00] (fourth10) {};
		\node [below = .35 of fourth10] (fourth10winningSet) {\small \textcolor{red}{$(A\times B)^\omega$}};
		\node[history] [right = 1.5 of fourth10] (fourth11) {};
		\node [below = .35 of fourth11] (fourth11winningSet) {\small \textcolor{red}{$(A\times B)^\omega$}};
		
		\node [below left = 1.5 and 2.25 of fourth00] (etc) {...};
		
		\path (epsilon) edge node [above left] {$00$} (00);
		\path (epsilon) edge node [left] {$01$} (01);
		\path (epsilon) edge node [left] {$10$} (10);
		\path (epsilon) edge node [left = .08] {11} (11);
		
		\path (00) edge node [above left] {$00$} (second00);
		\path (00) edge node [left] {$01$} (second01);
		\path (00) edge node [left] {$10$} (second10);
		\path (00) edge node [left = .08] {$11$} (second11);
		
		\path (second00) edge node [above left] {$00$} (third00);
		\path (second00) edge node [left] {$01$} (third01);
		\path (second00) edge node [left] {$10$} (third10);
		\path (second00) edge node [left =.08] {$11$} (third11);
		
		\path (third00) edge node [above left] {$00$} (fourth00);
		\path (third00) edge node [left] {$01$} (fourth01);
		\path (third00) edge node [left] {$10$} (fourth10);
		\path (third00) edge node [left = .08] {$11$} (fourth11);
		
		\path (fourth00) edge node [above left] {$00$} (etc);
		
		\end{tikzpicture}
		\caption{A game for which the naive algorithm does not work. For the sake of concision, the pairs of actions in $A\times B$ are written as two-letter words. In red is the partial winning set of the corresponding history.}
		\label{FIG:naiveNotEnough}
	\end{center}
\end{figure}
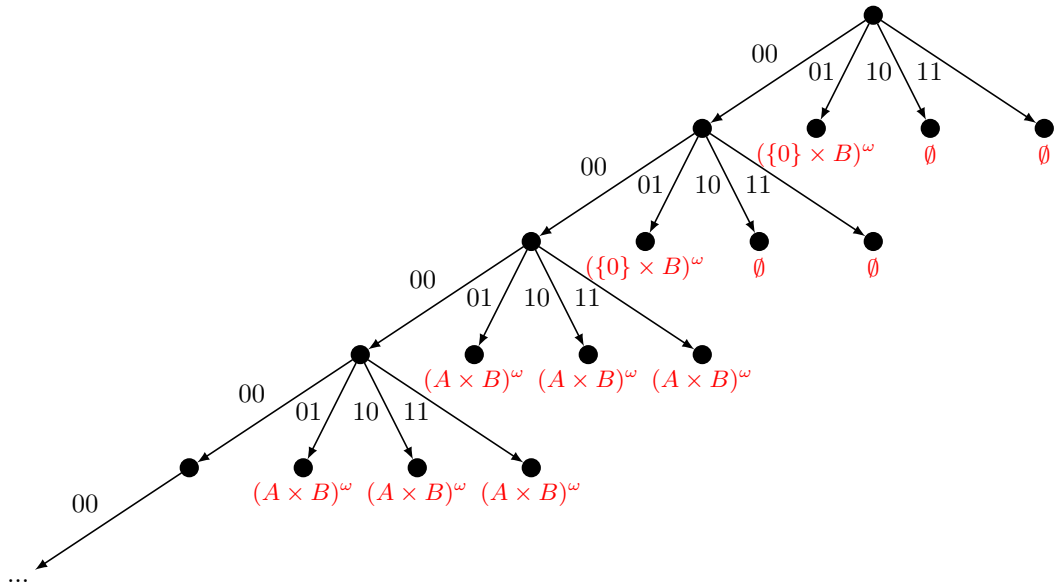

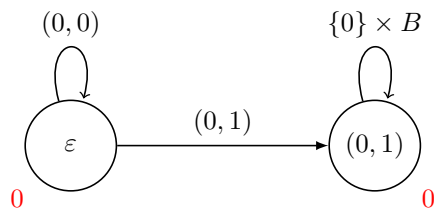
\begin{figure}
	\begin{center}
		\begin{tikzpicture}[-latex, auto, node distance = 1.5 cm and 1.5 cm, on grid, semithick, trianglenode/.style={draw, triangle, color = red}, history/.style={draw, circle, fill, scale = .7}, state/.style={draw, circle, minimum size  = 1.2cm}, every text node part/.style={align=center}]
		
		\node[state] (epsilon) {$\varepsilon$};
		\node[state] [right = 4 of epsilon] (01) {$(0,1)$};
		\node (epsilonAction) [below left = 1 of epsilon] {\textcolor{red}{0}};
		\node (01Action) [below right = 1 of 01] {\textcolor{red}{0}}; 
		
		\path (epsilon) edge [loop above] node[above] {$(0,0)$} (epsilon);
		\path (01) edge [loop above] node[above] {$\{0\}\times B$} (01);
		\path (epsilon) edge node[above] {$(0,1)$} (01);
		\end{tikzpicture}
		\caption{The finite-memory strategy generated by approaches $(*)$ and $(**)$ for the game represented in Figure \ref{FIG:naiveNotEnough}. Each state is labeled by the history associated to it, and in red is the action associated to that state.}
		\label{FIG:naiveNotEnoughStrategy}
	\end{center}
\end{figure}

\begin{example}\label{EX:naiveNotEnough}
	Consider the game $(A,B,W)$ with $A = B = \{0,1\}$ and $W = (0,0)^*(0,1)(\{0\}\times B)^\omega + (0,0)^2(0,0)^*(\{1\}\times B + A\times \{1\})(A\times B)^\omega$. This game is described in Figure \ref{FIG:naiveNotEnough}. In other words, Player $1$ has two ways to win:
	\begin{itemize}
		\item either the players play $(0,1)$ or $(0,0)(0,1)$ and Player $1$ then only has to play action $0$ forever;
		\item or the players play $(0,0)$ twice and then one player plays action $1$, reaching a point where all possible continuations are winning for Player $1$.
	\end{itemize}
	As $W$ can be expressed via a regular expression, it induces finitely many partial winning sets, which means that $\subseteq$ trivially induces a well partial order over said partial winning sets. Moreover, $W$ can be expressed as the union of an open set and a closed set (the closed set corresponds to the first item above, while the open set corresponds to the second item), and hence belongs to the Hausdorff difference hierarchy (more precisely it belongs to $K_2$).
	
	Moreover, one can easily check that $W_\varepsilon \subseteq W_{(0,0)}$. As a consequence, both approach $(*)$ and approach $(**)$ yield the finite-memory strategy described in Figure \ref{FIG:naiveNotEnoughStrategy}. This strategy is not winning for Player $1$, as if Player $2$ always plays action $0$ it will generate the play $(0,0)^\omega$, which does not belong to $W$.
\end{example}

\subsection{Proofs for sets in $K_2$}

To better understand how the proof works in the general case, we propose here to study the basic case of sets in $K_2$.
As we have seen previously, the two 
approaches that worked well for the case where the winning set is a closed set do not suffice in that case. We prove the following result:


\begin{theorem}\label{THM:K2}
       If $W$ is in $K_2$,  if  $\subseteq$ induces a well partial order on $(W_h)_{h \in \Gamma}$  and if $\varepsilon \in \Gamma$, then Player $1$ has a finite-memory winning strategy from $\varepsilon$.
\end{theorem}

Let us suppose then that $W$ is in $K_2$: as already mentioned, $W$ is the union of a closed set $C$ and an open set $O$. We let $\mathcal{F}$ be the generating family of $O$ and denote by $\PrefC$ the set of histories which have at least one continuation in $C$, that is $\PrefC = \{h\in\mathcal{H} \mid \exists \rho, h\rho\in C\}$. 

As a preliminary observation, recall we already know how to handle the case when the winning set is open. The method also works well for the general case when Player $1$ is able to reach $O$ by herself (that is, she have a winning strategy for $O$). We also know that finding a finite-memory non-losing strategy for Player $1$ is always possible (see for instance the method $(*)$ for the case where the winning set is closed). As a consequence, a simple method one would be tempted to try would be the following:
\begin{itemize}
	\item follow some non-losing strategy as long as the current history belongs to $\PrefC$;
	\item as soon as we detect we have left $\PrefC$, play some finite-memory winning strategy to reach a history in $\mathcal{F}$ (this is possible because if we have played in a non-losing fashion so far and the current history does not belong to $\PrefC$, then the only way to win from there is to produce a play that belongs to $O$).
\end{itemize}
This method should produce a finite-memory winning strategy, however it relies on the assumption that one is able to detect whether or not the current history belongs to $\PrefC$ using only finite memory. This assumption does not rely on any solid ground, which makes this method incorrect. We propose another construction of a finite-memory winning strategy, which does not need to detect when the current history stops belonging to $\PrefC$, but which ensures that $\mathcal{F}$ will be reached if it were the case (despite not knowing it). 

Consider indeed a history $\hpref$ in $\Gamma\cap\PrefC$ and a history $h$ in $\Gamma$ such that $h\notin\PrefC$ and $W_{\hpref}\subseteq W_h$. We call $TC(\hpref, h)$ the set $\{hl \mid \hpref l \in \PrefC\}$ consisting of the finite continuations from $\hpref$ that belong to $\PrefC$, but rooted in $h$. Notice that for all infinite continuations $\rho$ such that $\hpref\rho\in C$, we have $\hpref\rho\in W$, which means that $h\rho\in W$ (since $W_{\hpref}\subseteq W_h$) and hence that $h\rho\in O$ since $h\notin\PrefC$ (which means that all winning continuations of $h$ belong to $O$ because they cannot belong to $C$). We thus know that $TC(\hpref, h)$ contains a family of histories $\mathcal{F}_{\hpref, h}$ included in $\mathcal{F}$ (or in the case where $h$ itself has a strict prefix in $\mathcal{F}$, we set $\mathcal{F}_{\hpref, h} = \{h\}$) and such that all infinite branches of $TC(\hpref, h)$ have a finite prefix in $\mathcal{F}_{\hpref, h}$, and by Kőnig's lemma we know this family is finite. We call $\depth(\hpref, h)$ the maximal length of the elements in $\mathcal{F}_{\hpref,h}$. Intuitively, this means that, if the current history were $h$ but Player $1$ only knew of its under-approximation $\hpref$, she could ensure the win by following a play whose finite prefixes $l$ were such that $\hpref l \in\PrefC$ for $\depth(\hpref,h)$ steps. This however still requires to compute the value of $\depth(\hpref,h)$ and hence to know of $h$. However, as stated by the following lemma, the value of $\depth(\hpref,h)$ for all eligible $h$ is bounded.

\begin{lemma}\label{LM:depthBounded}
	For all $\hpref$ in $\PrefC$, $\{\depth(\hpref, h) \mid h\in \Gamma, h\notin \PrefC, W_{\hpref}\subseteq W_h\}$ is bounded.
\end{lemma}

The proof of this lemma can be found in appendix \ref{AP:caseK2}, and makes use of the well partial order hypothesis.

For $\hpref$ in $\Gamma\cap\PrefC$, we will then call $\depth(\hpref)$ the upper bound of $\depth(\hpref, h)$ for $h$ meeting the criteria described above. The idea is the following: if the current history is $h\in\Gamma$, but Player $1$ only knows of its under-approximation $\hpref$, and then plays some finite continuation $l$ of length $\depth(\hpref)$ (which is independent of $h$) such that $\hpref l\in \PrefC$, then she ensured the win as $hl$ has a finite prefix in $\mathcal{F}$. This is formally stated in the following lemma:

\begin{lemma}\label{followsCharpWinning}
	Let $\hpref\in\PrefC$ and $h\notin\PrefC$ such that $W_{\hpref}\subseteq W_h$. Let $\rho\in (A\times B)^\omega$ such that for all finite prefixes $l$ of $\rho$ such that $|l|\leq \depth(\hpref)$ we have $\hpref l\in \PrefC$. Then $h\rho\in O$.
\end{lemma}

\begin{proof}
	Let $l$ be the finite prefix of $\rho$ of length $\depth(\hpref)$. As $\depth(\hpref,h)\leq\depth(\hpref)$ we know that $hl$ has a prefix in $\mathcal{F}$, hence the result.
\end{proof}

Consider now a finite family $(h_i)_{i\in I}$ of histories in $\Gamma\setminus\PrefC$ such that for all $h\in\PrefC$ there exists $i\in I$ such that $W_{h_i}\subseteq W_h$. For all $i\in I$ there exists a finite-memory decision machine $(M_i,\sigma_i,\mu_i,m_i)$ associated with a of finite-memory strategy $(s_i)_{i\in I}$ such that $s_i$ wins from $h_i$. As a consequence, for all $h\in\Gamma\setminus\PrefC$ there exists $i\in I$ such that $s_i$ wins from $h$. Consider also a finite family $(\hpref_j)_{j\in J}$ of histories in $\Gamma\cap\PrefC$ indexed by $J\subseteq\mathbb{N}$ such that for all histories $\hpref$ in $\PrefC\cap\Gamma$ there exists $j\in J$ such that $W_{\hpref_j}\subseteq W_{\hpref}$. For all $j\in J$, let $T_j = \{\hpref_j l \mid |l|\leq\depth(\hpref_j)\}$. Up to renaming, we can suppose that the $T_j$'s are disjoint from one another. We build our finite-memory winning strategy $s = (M,\sigma,\mu,m_0)$ in the following way:
\begin{itemize}
	\item $M = \cup_{i\in I} M_i \cup \cup_{j\in J} T_j$;
	\item for $m\in M_i$ we let $\sigma(m) = \sigma_i (m)$;
	\item for $t\in T_j$ we let $\sigma(t)$ be any non-losing action from $t$;
	\item for $m\in M_i$ and $(a,b)\in A\times B$ we let $\mu(m, (a,b)) = \mu_i(m,(a,b))$;
	\item for $t\in T_j$ and $(a,b)\in A\times B$ such that $a = \sigma(t)$:
	\begin{itemize}
		\item if $t(a,b)\in T_j$ then $\mu(t,(a,b)) = t(a,b)$;
		\item else if $t(a,b)\in\Gamma\setminus\PrefC$ then there exists $i\in I$ such that $s_i$ wins from $t(a,b)$: we let $\mu(t,(a,b))=m_i$;
		\item else if $t(a,b)\in\Gamma\cap\PrefC$ then there exists $j\in J$ such that $W_{\hpref_j}\subseteq W_{t(a,b)}$ and we let $\mu(h,(a,b))=\hpref_j$;
	\end{itemize}
	\item $m_0 = \hpref_j$ where $j\in J$ is such that $W_{\hpref_j}\subseteq W_\varepsilon$.
\end{itemize}

We prove this finite-memory strategy is winning for Player $1$. To this end, let us first show that for all compatible histories $h$ such that $\mu(h)\in T_j$ for some $j\in J$, the memory state $\mu(h)$ provides an under-approximation of the winning set induced by $h$:

\begin{lemma}\label{memContainedReal}
	If $\mu(h)\in T_j$ for some $j\in J$ then we have $W_{\mu(h)}\subseteq W_h$.
\end{lemma}

\begin{proof}
	The proof is by induction on $h$. First we have $\mu(\varepsilon) = m_0$ and per the definition $W_{m_0}\subseteq W_\varepsilon$. Now consider $h\in \mathcal{H}$ such that $\mu(h) \in T_j$ for some $j\in J$ and $W_{\mu(h)}\subseteq W_h$. Let $(a,b)\in A\times B$ such that $\mu(h(a,b))\in T_{j'}$ for some $j'\in J$. Then,
	\begin{itemize}
		\item if $\mu(h)(a,b)\in T_j$ then we have $\mu(h(a,b)) = \mu(h)(a,b)$ and the desired result follows by Lemma \ref{orderTransitive},
		\item else we must have $\mu(h)(a,b)\in\Gamma\cap\PrefC$ (else we would not have $\mu(h(a,b))\in T_{j'}$) and $\mu(h(a,b))=\hpref_{j'}$ with $j'$ such that $W_{\hpref_{j'}}\subseteq W_{\mu(h)(a,b)}$, and $W_{\mu(h)(a,b)}\subseteq W_{h(a,b)}$ once again by Lemma \ref{orderTransitive}, which ensures the result.
	\end{itemize}
\end{proof}

As a consequence of Lemma \ref{memContainedReal}, when $h$ is such that $\mu(h)\in T_j$ for some $j\in J$ then we have $W_{\mu(h)}\subseteq W_h$. As a consequence, for all $(a,b)\in A\times B$ we have $W_{\mu(h)(a,b)}\subseteq W_{h(a,b)}$. This ensures that if $\mu(h)(a,b)\in\Gamma\setminus\PrefC$ and the strategy $s_i$ wins from $\mu(h)(a,b)$ then $s_i$ also wins from $h(a,b)$, which explains why our finite-memory strategy is winning. Formally, we have the following lemma:

\begin{lemma}\label{Miwins}
	If $\rho\in (A\times B)^\omega$ is compatible with $s$ and there exists $k\in\mathbb{N}$ and $i\in I$ such that $\mu(\rho_{\leq k}) = m_i$ then $\rho\in W$.
\end{lemma}

\begin{proof}
	Without loss of generality, we can suppose that $k$ is the smallest integer such that $\mu(\rho_{\leq k})\notin \cup_{j\in J} T_j$. We want to prove there exists a history $h$ such that $W_h\subseteq W_{\rho_{\leq k}}$ and $s_i$ is winning from $h$, as this will yield the desired result. Our candidate for $h$ is $\mu(\rho_{\leq k-1})(a,b)$ where $(a,b) = \rho_k$.
	\begin{itemize}
		\item By Lemma \ref{memContainedReal} we know that $W_{\mu(\rho_{\leq k-1})} \subseteq W_{\rho_{\leq k-1}}$, and thus $W_{\mu(\rho_{\leq k-1})(a,b)} \subseteq W_{\rho_{\leq k-1}}$.
		\item Furthermore, we know per the definition of $s$ that $s_i$ is winning from $\mu(\rho_{\leq k - 1}(a,b))$ since $\mu(\rho_{\leq k}) = m_i$.
	\end{itemize}
	Finally, a trivial induction shows that for any $\beta\in B^\omega$ we have $\out(\rho_{\leq k},s,\beta)=\out(\rho_{\leq k},s_i,\beta)$, and since $s_i$ is winning from $\rho_{\leq k}$ we have $\rho\in W$.
\end{proof}

Finally, with the help of Lemma \ref{memContainedReal} and Lemma \ref{Miwins}, we can prove the following, which concludes the proof of Theorem~\ref{THM:K2}.

\begin{lemma}
	$s$ is winning from $\varepsilon$.
\end{lemma}

\begin{proof}
	Let us consider $\rho\in (A\times B)^\omega$ compatible with $s$. We want to show that $\rho\in W$. If there exists $k\in\mathbb{N}$ such that $\mu(\rho_{\leq k})\in\cup_{i\in I}M_i$ then Lemma  \ref{Miwins} suffices to conclude. Suppose then that for all $k\in\mathbb{N}$ we have $\mu(\rho_{\leq k})\in\cup_{j\in J}T_j$. If $\rho\in C$ then obviously we have $\rho\in W$, so let us suppose that there exists $n_0\in\mathbb{N}$ such that $\rho_{\leq n_0}\notin \PrefC$. Due to the construction of $s$, there exists some $n_1\geq n_0$ and some $j\in J$ such that $\mu(\rho_{\leq n_1})=\hpref_j$. By Lemma \ref{memContainedReal} we then have $W_{\mu(\rho_{\leq n_1})}\subseteq W_{\rho_{\leq n_1}}$. Notice that we also have $\rho_{\leq n_1}\notin\PrefC$. Finally, let $l$ be the prefix of $\rho_{>n_1}$ of length $\depth(\hpref_j)$. By construction and since we do not have $\mu(\rho_{\leq k})\in\cup_{i\in I}M_i$ for any $k$, for all $k\leq \depth(\hpref_j)$ we have $\mu(\rho_{\leq n_1+k})=\hpref_j l_{\leq k}\in\PrefC$, and hence by application of Lemma \ref{followsCharpWinning} we can conclude that $\rho\in O$ and thus $\rho\in W$.
\end{proof}

\begin{figure}
	\begin{center}
		\begin{tikzpicture}[-latex, auto, node distance = 1.5 cm and 1.5 cm, on grid, semithick, trianglenode/.style={draw, triangle, color = red}, history/.style={draw, circle, fill, scale = .7}, state/.style={draw, circle, minimum size  = 1.2cm}, every text node part/.style={align=center}]
		
		\node (epsilon) {$\varepsilon$};
		\node [below left = 1 and 2.5 of epsilon] (00) {$(0,0)$};
		\node [below right = 1 and 2.5 of epsilon] (01) {$(0,1)$};
		\node[below right = of 00] (0001) {$(0,0)(0,1)$};
		\node[below left = of 01] (0100) {$(0,1)(0,0)$};
		\node[below right = of 01] (0101) {$(0,1)(0,1)$};
		
		\node [right = 8 of epsilon] (01bis) {$(0,1)$};
		
		\node [below right = 5 and 4 of epsilon] (m) {$m$};
		
		\path (epsilon) edge node[above left] {\small $(0,0)$} (00);
		\path (epsilon) edge node[above right] {\small $(0,1)$} (01);
		\path (00) edge node[above right] {\small $(0,1)$} ( 0001);
		\path (01) edge node[above left] {\small $(0,0)$} (0100);
		\path (01) edge node[above right] {\small $(0,1)$} (0101);
		\path (00.south) edge [color = red, bend right=48] node[below left] {$(0,0)$} (m);
		\path (01bis) edge [loop right] node[right] {\small $\{0\}\times B$} (01bis);
		\path (0001.south) edge [color = red, bend right=58] node[below right] {$\{0\}\times B$} (01bis);
		\path (0100.south) edge [color = red, bend right=48] node[below left = .25 and 2] {$\{0\}\times B$} (01bis);
		\path (0101.south) edge [color = red, bend right] node[above left] {$\{0\}\times B$} (01bis);
		\path (m) edge [loop below] node[below] {$\{1\}\times B$} (m);
		\end{tikzpicture}
		\caption{The finite-memory strategy for the game represented in Figure \ref{FIG:naiveNotEnough}.}
		\label{FIG:FMStrategy}
	\end{center}
\end{figure}
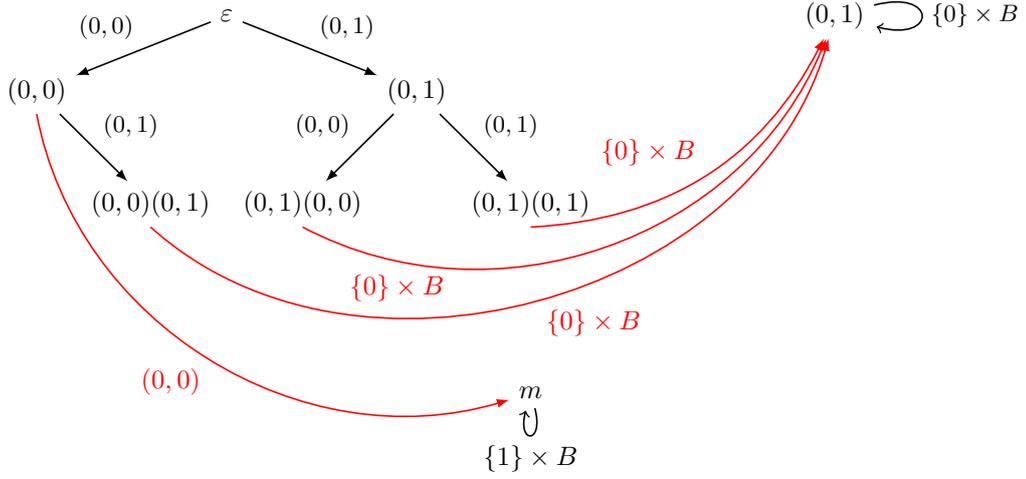

We illustrate our method on the game described in Figure \ref{FIG:naiveNotEnough}:

\begin{example}
	Consider once again the game represented in \ref{FIG:naiveNotEnough}. We recall here that we have $A = B = \{0,1\}$ and $W = (0,0)^*(0,1)(\{0\}\times B)^\omega + (0,0)^2(0,0)^*(\{1\}\times B + A\times \{1\})(A\times B)^\omega$. We let $C = (0,1)(\{0\}\times B)^\omega + (0,0)(0,1)(\{0\}\times B)^\omega$ and $O = (0,0)^2(0,0)^*[(0,1) + (1,0) + (1,1)](A\times B)^\omega$, and we have $W = C\cup O$. One can check easily that $C$ is a closed set and $O$ is an open set generated by the family of histories $\mathcal{F} = (0,0)^2(0,0)^*[(0,1) + (1,0) + (1,1)]$.
	
	The histories in $\PrefC$ are of three different types: $\varepsilon$, whose induced winning set is $W$, $(0,0)$, whose induced winning set is included in $W$, and all the other histories in $\PrefC$, whose induced winning set is $(\{0\}\times B)^\omega$. We choose two histories in $\PrefC$ that provide under-approximations for the open sets induced by all elements of $\PrefC$: $\varepsilon$ and $(0,1)$. Except for histories which already have a prefix in $\mathcal{F}$, no history out of $\PrefC$ induces a winning set that includes $W_{(0,1)}$, which means that $\depth((0,1)) = 0$. However $W_\varepsilon$ is included in all winning sets induced by the histories in $(0,0)^2(0,0)^*$. Since $(0,0)^2(0,0)^*(0,0)(0,1)\subseteq\mathcal{F}$, we have $\depth(\varepsilon) = 2$. The only continuation of length $2$ from $\varepsilon$ that goes out of $\PrefC$ is $(0,0)(0,0)$, and the finite-memory strategy we chose that reaches $O$ from $(0,0)(0,0)$ is one where Player $1$ always plays action $1$.
\end{example}

\subsection{Proof for sets in $\Lambda_\theta$}

We study now the case where $W\in\Lambda_\theta$ for an ordinal $\theta> 1$ and Theorem \ref{thebigtheorem} is true for all winning sets belonging to $\Lambda_\eta$ or $K_\eta$ for all $\eta<\theta$. As always, we suppose that Player $1$ has a winning strategy from $\varepsilon$ and we want to show she has a finite-memory winning strategy.

As $W\in\Lambda_\theta$, there exists a family of sets $(S_i)_{i\in I}$ such that $W = \Cup_{i\in I} S_i$ and for each $i$ there exists $\theta_i < \theta$ such that $S_i\in\Lambda_{\theta_i}$ or $S_i\in K_{\theta_i}$. Furthermore, there exists a disjoint family of open sets $(O_i)_{i\in I}$ such that for each $i\in I$ we have $S_i\subseteq O_i$. Finally, for each $i$ the set $O_i$ is generated by a family of histories $\mathcal{F}_i$.

Consider then a winning strategy $s$ for Player $1$ and let $T_s$ be its induced strategic tree. Consider the tree $T_s' = T_s\setminus \{h\in\mathcal{H}\mid \exists i\in I, \exists f\in \mathcal{F}_i, f\sqsubset h\}$. Since $s$ is winning, all infinite branches of $T_s$ belongs to $O_i$ for some $i$, and hence $T_s'$ does not have any infinite branch. By Kőnig's lemma, this means that $T_s'$ is a finite tree. Let us consider the set $\mathcal{L}$ of maximal elements (with regards to the prefix relation) in $T_s'$. By construction all histories $h$ in $\mathcal{L}$ are such that there exists a unique (because the $O_i$'s are disjoint from one another) $i\in I$ such that $h\in \mathcal{F}_i$. This means that the winning plays which have $h$ as a prefix are included in $S_i$ and hence for $h\in\mathcal{L}$ we have $W_h = h^{-1}(S_i\cap \cyl(h))$. Thus by Lemma \ref{LM:closedIntersectionCylinder} (which states that both $\Lambda_{\theta_i}$ and $K_{\theta_i}$ are closed under intersection with a cylinder, see appendix \ref{AP:fineHausdorff}) we know that $W_h\in \Lambda_{\theta_i}$ or $W_h\in K_{\theta_i}$ (depending on whether $S_i$ belongs to $\Lambda_{\theta_i}$ or $K_{\theta_i}$). Moreover, since all histories in $\mathcal{L}$ belong to $T_s$ and $s$ is winning we know that $\mathcal{L}\subseteq \Gamma$. This means that we have all the hypotheses we need to apply the induction hypothesis to any $l\in \mathcal{L}$ (namely, $l\in\Gamma$ and $W_l\in \Lambda_{\eta}$ or $W_l\in K_{\eta}$ for some $\eta<\theta$), and hence for all $l\in \mathcal{L}$ there exists a finite-memory strategy $s_l = (M_l,\sigma_l,\mu_l,m_l)$ that wins from $l$ (up to renaming, we suppose that the $M_l$'s are disjoint from one another and from $T_s$).

We will build a finite-memory strategy $s_f = (M,\sigma,\mu,m_0)$ for Player $1$ in the following way:
\begin{itemize}
	\item $M = T_s'\cup \uplus_{l\in\mathcal{L}}M_l$;
	\item $\sigma(m) = s(m)$ for $m\in T_s'$, and $\sigma(m) = \sigma_l(m)$ for $m\in M_l$;
	\item $\mu(m,(a,b)) = m(a,b)$ if $m\in T_s'$ and $m(a,b)\neq l$ for all $l\in\mathcal{L}$ (where $a=\sigma(m)$);
	\item$\mu(m,(a,b)) = m_l$ if $m\in T_s'$ and $m(a,b) = l$ for some $l\in\mathcal{L}$ (where $a = \sigma(m)$);
	\item $\mu(m,(a,b)) = \mu_l(m(a,b))$ if $m\in M_l$;
	\item $m_0 = \varepsilon$ if $\varepsilon \neq l$ for all $l\in \mathcal{L}$, and $m_0 = m_l$ if there exists $l\in\mathcal{L}$ such that $l = \varepsilon$.	
\end{itemize}

The idea behind this construction is the following: we follow the winning strategy $s$ until we reach some $l\in \mathcal{L}$, from which we know the strategy $s_l$ is winning. It remains only to emulate $s_l$ from this point onwards to guarantee the win. The formal proof that this strategy is winning can be found in Appendix \ref{AP:inductiveCaseLambda}.

\subsection{Proof for sets in $K_\theta$}

As this induction step is very technical, we do not present it here, but provide a detailed proof in appendix \ref{AP:inductiveCaseK}.

\section{Tightness of the result}\label{sec:tightness}

We explore the tightness of our result. In particular, the winning sets of the two games in Example~\ref{ex:tightness} below are in $\Pi_2^0$ and $\Sigma_2^0$, respectively, just above $\Delta^0_2$ in the Borel hierarchy; the two games satisfy the well partial order assumption and Player 1 has winning strategies, but no finite-memory winning strategies. Note that the example in $\Pi_2^0$ is harder to define and deal with than the one in $\Sigma_2^0$. (See details in appendix \ref{AP:tightness}.)

\begin{example}\label{ex:tightness}
		For the counter-example in $\Pi_2^0$, let $A$ be a finite set of at least two elements and $w$ a disjunctive sequence on $A\times\{0\}$. We define the labeling function $l: (A\times\{0\})^*\to \{0,1\}$ in the following way: $l(h) = 1$ if and only if there exists $h_0$ and $h'$ such that $h = h_0h'$ and $h'$ is the longest factor of $h$ that is also a prefix of $w$. Let then $W$ be the set defined by $\rho\in W$ if and only if infinitely many prefixes $h$ of $\rho$ are such that $l(h) = 1$ and consider the game $(A,\{0\},W)$.
		
		For the counter-example in $\Sigma_2^0$, let $A$ be a finite set of at least two elements and $w$ an irregular sequence in $(A\times\{0\})^*$. Let $W = \{\rho\in (A\times \{0\})^\omega\mid\exists h\in (A\times \{0\})^*, \rho = h\rho_0 \text{ where } \rho_0 \text{ is a suffix of } w\}$. $W$ is the set of sequences which have a suffix in common with $w$. Consider then the game $(A,\{0\}, W)$.
\end{example}

We also considered a relaxation of the well partial order assumption, but found a counter-example with a closed winning set. (See details in appendix \ref{AP:tightness}.)

Finally, we studied the statement of our result where ``Player 1'' is syntactically replaced with ``Player 2'', but the $W_h$ still pertain to Player 1. The variant holds for closed sets, but we found a counter-example just above in the Hausdorff difference hierarchy.

\section{Conclusion}\label{sec:conclusion}

To conclude, we have proven the existence of finite-memory winning strategies under certain conditions on the winning set for Player $1$. These conditions are met for well studied games such as energy games \cite{energyGames} or games where the winning condition is a Boolean combination of reachability and safety objectives, which makes our result a generalization of known results on the topic. This result relies on contributions to the field of descriptive set theory, in particular on representation of sets in $\Delta_2^0$.

We have also studied the tightness of our result. The results we currently have in this direction encourage us to think that our hypotheses are tight and that weakening them is no easy task. In the future, we want to extend these tightness results by exploring other possible hypotheses, as well as study infinitely branching games, when the action sets of the players are not finite.



\bibliography{csl_article}

\appendix

\section{Descriptive set theory}\label{AP:descriptiveSetTheory}

\subsection{The Hausdorff difference hierarchy}

In this section, we provide properties on the Hausdorff difference hierarchy. We begin by showing that it is indeed an increasing hierarchy:

\begin{lemma}\label{HausdorffHierarchyIncreasing}
	For all ordinals $\theta_1$, $\theta_2$ such that $\theta_1 < \theta_2$, we have $\mathcal{D}_{\theta_1}\subseteq \mathcal{D}_{\theta_2}$.
\end{lemma}

\begin{proof}
	We prove that for all ordinals $\theta$ we have $\mathcal{D}_\theta\subseteq \mathcal{D}_{\theta + 1}$, which ensures the result by a straightforward induction.
	
	Consider an ordinal $\theta$ and let $S\in \mathcal{D}_\theta$. Consider an increasing family of open sets $(O_\eta)_{\eta<\theta}$ such that $S = D_\theta((O_\eta)_{\eta<\theta})$, that is, $s\in S$ if and only if there exists $\eta$ such that $s\in O_\eta$ and the least such $\eta$ is of parity opposite to that of $\theta$. Let then $(P_\eta)_{\eta<\theta+1}$ be the increasing family of open sets defined by the following:
	\begin{itemize}
		\item $P_0 = \emptyset$;
		\item for all successor ordinals $\eta$, $P_\eta = O_{\eta -1}$;
		\item for all non-zero limit ordinals $\eta$, $P_\eta = \cup_{\zeta<\eta} O_\zeta$.
	\end{itemize}
	We will show that $S = D_{\theta+1}((P_\eta)_{\eta<\zeta})$. Consider $s\in C$. For all ordinals $\eta$, we have $s\in O_\eta \Leftrightarrow s\in P_{\eta + 1}$, which means that $s$ belongs to $S$ if and only if there exists $\eta<\theta$ such that it belongs to $P_\eta$ and the least such $\eta$ is of parity opposite to $\theta+1$, that is $s$ belongs to $S$ if and only if it belongs to $\mathcal{D}_{\theta + 1}((P_\eta)_{\eta<\theta+1})$.
\end{proof}

\begin{lemma}\label{HausdorffClosedIntersectionOpen}
	For all ordinals $\theta$, $\mathcal{D}_\theta$ is closed under intersection with an open set.
\end{lemma}

\begin{proof}
	Let $\theta$ be an ordinal and $S\in \mathcal{D}_\theta$. There exists a family of open sets $(O_\eta)_{\eta<\theta}$ such that $S = D_\theta((O_\eta)_{\eta<\theta})$. Let $O$ be an open set. We will prove that $S\cap O = D_\theta((O_\eta\cap O)_{\eta<\theta})$.
	
	Let $\rho\in S\cap O$. Obviously for all $\eta<\theta$ such that $\rho\in O_\eta$ we have $\rho\in O_\eta\cap O$. Hence $\rho\in D_\theta((O_\eta)_{\eta<\theta})$.
	
	Conversely, let $\rho\in D_\theta((O_\eta\cap O)_{\eta<\theta})$. There exists $\eta<\theta$ such that $\rho\in O_\eta \cap O$, and hence obviously $\rho\in O$. Then for all $\eta<\theta$ we have $\rho\in O_\eta \Leftrightarrow \rho\in O_\eta \cap O$, and hence $\rho\in S \Leftrightarrow \rho\in D_\theta((O_\eta)_{\eta<\theta}) \Leftrightarrow \rho\in D_\theta((O_\eta\cap O)_{\eta<\theta})$, which means that $\rho\in S$.
\end{proof}

\begin{lemma}\label{HausdorffComplementIncreases}
	For all ordinals $\theta$, if $S\in \mathcal{D}_\theta$ then $\overline{S}\in \mathcal{D}_{\theta+1}$.
\end{lemma}

\begin{proof}
	Let $\theta$ be an ordinal and $S\in\mathcal{D}_\theta$. There exists a family of open sets $(O_\eta)_{\eta<\theta}$ such that $S = D_\theta((O_\eta)_{\eta<\theta})$. Let $O_\theta = (A\times B)^\omega$. Then $\overline{S} = D_{\theta+1}((O_\eta)_{\eta<\theta+1})$.
\end{proof}

\begin{lemma}\label{DthetaSminusO}
	For all successor ordinals $\theta$, a set $S$ belongs to $\mathcal{D}_\theta$ if and only if there exists an open set $O$ and a set $S'$ in $\mathcal{D}_{\theta-1}$ such that $S = O\setminus S'$.
\end{lemma}

\begin{proof}
	Let $\theta$ be a successor ordinal, $S'\in\mathcal{D}_{\theta-1}$ and $O$ an open set. By Lemma \ref{HausdorffComplementIncreases}, $\overline{S'}\in\mathcal{D}_\theta$, and we have $O\setminus S' = O\cap\overline{S'}$ which by Lemma \ref{HausdorffClosedIntersectionOpen} means that $O\setminus S'\in \mathcal{D}_\theta$.
	
	Conversely, let $S\in\mathcal{D}_\theta$. There exists an increasing family of open sets $(O_\eta)_{\eta < \theta}$ such that $S = D_\theta((O_\eta)_{\eta<\theta})$. Hence a play $\rho$ belongs to $S$ if and only if it belongs to $O_{\theta-1}$ and the least $\eta$ such that $\rho\in O_\eta$ is of parity opposite to $\theta$, i.e. is not of parity opposite to $\theta-1$. That is, a play $\rho$ belongs to $S$ if and only if it belongs to $O_{\theta-1}\setminus D_{\theta-1}((O_\eta)_{\eta<\theta-1})$, hence $S = O_{\theta-1}\setminus D_{\theta-1}((O_\eta)_{\eta<\theta-1})$.
\end{proof}

\subsection{The fine Hausdorff hierarchy} \label{AP:fineHausdorff}

In this section we provide properties on the fine Hausdorff hierarchy that lead to different ways of defining the sets in the hierarchy. We begin once again by showing it is an increasing hierarchy:

\begin{lemma}\label{increasingHierarchy}
	For $\theta_0 < \theta_1$, we have $\Lambda_{\theta_0}\cup K_{\theta_0}\subseteq \Lambda_{\theta_1}$ and $\Lambda_{\theta_0}\cup K_{\theta_0}\subseteq K_{\theta_1}$.
\end{lemma}

\begin{proof}
	We have $\Lambda_{\theta_0}\cup K_{\theta_0}\subseteq \Lambda_{\theta_1}$ from the definition of $\Lambda_{\theta_1}$, as any set can be written as the open union of itself only. Moreover, since sets in $\Lambda_{\theta_0}$ and $K_{\theta_0}$ are the complements of each other, and since $K_{\theta_1}$ contains the sets whose complement is in $\Lambda_{\theta_1}$ we also have $\Lambda_{\theta_0}\cup K_{\theta_0}\subseteq K_{\theta_1}$.
\end{proof}

This yields two corollaries that provide other ways tosee the sets in $\Lambda_\theta$. The first one states that sets in $\Lambda_\theta$ can be written as open unions of sets in $K_\theta$, while the second one goes further and states that $\Lambda_\theta$ contains exactly the sets that can be written as open unions of sets that belong to $\cup_{\eta<\theta}K_\eta$.

\begin{corollary}\label{LambdaOpenUnionOfK}
	For every ordinal $\theta>0$, any set $S$ in $\Lambda_\theta$ can be written as the open union of sets in $K_\theta$.
\end{corollary}

\begin{proof}
	If $S$ belongs to $\Lambda_1$ then $S$ is an open set associated with a generating family $\mathcal{F}$. Without loss of generality, we can suppose that elements of $\mathcal{F}$ are incomparable with regards to the prefix relation, and we then have $S = \Cup_{f\in\mathcal{F}}\cyl(f)$, and for all $f$ $\cyl(f)$ is a closed set and hence belongs to $K_0$.
	
	Alternatively consider $S$ in $\Lambda_\theta$ with $\theta > 1$: there exists a family of sets $(S_i)_{i\in I}$ such that for each $S_i$ there exists an ordinal $\theta_i < \theta$ such that $S_i\in \Lambda_{\theta_i}$ or $S_i\in K_{\theta_i}$ and $S = \Cup_{i\in I} S_i$. By Lemma \ref{increasingHierarchy} we have that each $S_i$ belongs to $K_\theta$, thus yielding the result.
\end{proof}

\begin{corollary}\label{LambdaOpenUnionOfKSmaller}
	For all ordinals $\theta > 1$, a set $S$ belongs to $\Lambda_\theta$ if and only if there exists a family of sets $(S_i)_{i\in I}$ such that $S = \Cup_{i\in I} S_i$ such that for each $i$ there exists an ordinal $\theta_i<\theta$ such that $S_i\in K_{\theta_i}$.
\end{corollary}

\begin{proof}
	Consider $\theta > 1$ and a set $S$ in $\Lambda_\theta$. There exists a family of sets $(S_i)_{i\in I}$ such that $S = \Cup_{i\in I} S_i$ such that for each $i$ there exists an ordinal $\theta_i$ such that $S_i\in \Lambda_{\theta_i}$ or $S_i\in K_{\theta_i}$. Let us write $I = I_\Lambda \uplus I_K$ where $i\in I_\Lambda$ if and only if $S_i$ belongs to $\Lambda_{\theta_i}$ and not $K_{\theta_i}$, and $i\in I_K$ if and only if $S_i\in K_{\theta_i}$. For all $i\in I_\Lambda$, by application of Corollary \ref{LambdaOpenUnionOfK} there exists a family of sets $(S_i^j)_{j\in J_i}$ such that each $S_i^j$ belongs to $K_{\theta_i}$ and $S_i = \Cup_{j\in J_i} S_i^j$. We then have $S = \Cup_{i\in I_\Lambda} \Cup_{j\in J_i} S_i^j \Cup \Cup_{i\in I_K} S_i$, which proves the result.
\end{proof}

\begin{lemma}\label{LambdaClosedOpenUnion}
	For all ordinals $\theta$, $\Lambda_\theta$ is closed under open union.
\end{lemma}

\begin{proof}
	If $\theta = 1$ then the result is straightforward because open sets are closed under union.
	
	Consider an ordinal $\theta>1$ and a family $(S_i)_{i\in I}$ of sets in $\Lambda_\theta$ such that there exists a family of disjoint open sets $(O_i)_{i\in I}$ such that for all $i$ we have $S_i\subseteq O_i$. For each $i$ there exists a family of sets $(S_j^i)_{j\in J_i}$ such that $S_i = \Cup_{j\in J_i} S_j^i$. and for all $j\in J_i$ there exists $\theta_j < \theta$ such that $S_j^i\in \Lambda_{\theta_j}$ or $S_j^i\in K_{\theta_j}$ (up to renaming, we suppose that all the $J_i$'s are disjoint), and such that there exists a disjoint family of open sets $(O^i_j)_{j\in J_i}$ such that for all $j\in J_i$ we have $S_j^i\subseteq O_j^i$. We then have $S = \cup_{i\in I, j\in J_i} S_j^i$ where for each $i$ and for each $j\in J_i$ we have $S^i_j\in \Lambda_{\theta_j}$ or $S^i_j\in K_{\theta_j}$ and $S^i_j\subseteq O^i_j\cap O_i$, which means that $S$ belongs to $\Lambda_\theta$.
\end{proof}

\begin{lemma}\label{LM:closedIntersectionCylinder}
	For all ordinals $\theta$, both $\Lambda_\theta$ and $K_\theta$ are closed under intersection with a cylinder: for all $S$ in $\Lambda_\theta$ (resp. $K_\theta$), and for all $h\in \mathcal{H}$ we have $\cyl(h)\cap S \in \Lambda_\theta$ (resp. $K_\theta$).
\end{lemma}

\begin{proof}
	We prove the result by induction. First, since open sets are closed under finite intersection, we obviously have the result for $\Lambda_1$. Consider now an ordinal $\theta > 1$, and suppose that the result is true for all $\theta' < \theta$. Consider some $S$ in $\Lambda_\theta$ and some $h\in\mathcal{H}$. There exists a family of sets $(S_i)_{i\in I}$ such that $S = \Cup_{i\in I} S_i$ and for all $i\in I$ there exists $\theta_i < \theta$ such that $S_i\in \Lambda_{\theta_i}$ or $S_i\in K_{\theta_i}$. We then have $\cyl(h)\cap S = \Cup_{i\in I} \cyl(h)\cap S_i$, and by the induction hypothesis we can conclude $\cyl(h)\cap S\in\Lambda_{\theta}$. Consider now some $S$ in $K_\theta$ and some $h\in\mathcal{H}$. We have $\overline{\cyl(h)\cap S} = \overline{\cyl(h)} \cup \cyl(h)\cap\overline S$. Since $\overline{\cyl(h)}$ is an open set and hence belong to $\Lambda_1\subseteq \Lambda_\theta$, and $\cyl(h)\cap\overline S$ belongs to $\Lambda_\theta$ by what we have just shown, and both sets are contained in disjoint open sets $\overline{\cyl(h)}$ and $\cyl{h}$, by Lemma \ref{LambdaClosedOpenUnion} we can conclude that $\overline{\cyl(h)\cap S}\in \Lambda_\theta$, and hence $\cyl(h)\cap S\in K_\theta$, which concludes the induction step.
\end{proof}

\begin{corollary}\label{LambdaClosedIntersectionOpen}
	For all ordinals $\theta$, $\Lambda_\theta$ is closed under intersection with an open set: for all $S$ in $\Lambda_\theta$ and all open sets $O$ we have $O\cap S\in \Lambda_\theta$.
\end{corollary}

\begin{proof}
	Consider a set $S$ in $\Lambda_\theta$ and an open set $O$ associated with a generating family of histories $\mathcal{F}$. Without loss of generality we can suppose that no two distinct elements in $\mathcal{F}$ are comparable with regards to the prefix relation. Notice then that $O\cap S = \Cup_{f\in\mathcal{F}}\cyl(f)\cap S$. By Lemma \ref{LM:closedIntersectionCylinder}, for all $f$ we have that $\cyl(f)\cap S$ belongs to $\Lambda_\theta$, which means that by Lemma \ref{LambdaClosedOpenUnion} $O\cap S$ belongs to $\Lambda_\theta$ as the open union of sets in $\Lambda_\theta$.
\end{proof}

\begin{lemma}
	For any successor ordinal $\theta>1$, a set $S$ belongs to $K_\theta$ if and only if there exists a closed set $C$ and a set $L$ belonging to $\Lambda_{\theta-1}$ such that $S = C\cup L$.
\end{lemma}

\begin{proof}
	Suppose that $S$ belongs to $K_\theta$. Then there exists a family of sets $(S_i)_{i\in I}$ such that each $S_i$ belongs to either $\Lambda_{\theta-1}$ or $K_{\theta-1}$ and $\overline{S} = \Cup_{i\in I} S_i$. Thanks to Corollary \ref{LambdaOpenUnionOfKSmaller}, we can suppose that for all $i$ $S_i$ belongs to $K_{\theta-1}$. Consider a disjoint family of open sets $(O_i)_{i\in I}$, such that for all $i\in I$ we have $S_i\subseteq O_i$. Let $O = \cup_{i\in I} O_i$ (notice $O$ is an open set) and $C$ be the complement of $O$. We have $S = C\cup \Cup_{i\in I} (O_i\setminus S_i) = C\cup \Cup_{i\in I}  (O_i\cap \overline{S_i})$. For all $i$ $\overline{S_i}$ belongs to $\Lambda_{\theta-1}$, which means that by Corollary \ref{LambdaClosedIntersectionOpen} for all $f\in\mathcal{F}$ $\cyl(f)\cap\overline{S_i}$ belongs to $\Lambda_{\theta-1}$ as well. Finally, by Lemma \ref{LambdaClosedOpenUnion} we have that $\Cup_{i\in I} O_i\cap \overline{S_i}$ belongs to $\Lambda_{\theta-1}$ as the open union of sets in $\Lambda_{\theta-1}$.	
	
	Now suppose that $S$ is such that there exists a closed set $C$ and a set $L$ belonging to $\Lambda_{\theta-1}$ such that $S = C\cup L$. We have $\overline{S} = \overline{C}\cap \overline L$, where $\overline{C}$ is an open set and $\overline{L}$ belongs to $K_{\theta-1}$ and hence to $\Lambda_{\theta}$. By application of Corollary \ref{LambdaClosedIntersectionOpen}, $\overline{S}$ belongs to $\Lambda_\theta$ and hence $S$ belongs to $K_\theta$.
\end{proof}

\begin{lemma}\label{LM:KUnionCLambda}
	For any ordinal $\theta > 0$, if a set $S$ belongs to $K_\theta$ then there exists a closed set $C$ and a set $L$ belonging to $\Lambda_{\theta}$ such that $S = C\cup L$.
\end{lemma}

\begin{proof}
	Consider an ordinal $\theta > 0$ and let $S\in K_\theta$. We have $\overline{S}\in\Lambda_\theta$, which by Corollary \ref{LambdaOpenUnionOfKSmaller} means that there exists a family of sets $(S_i)_{i\in I}$ such that $\overline{S} = \Cup S_i$ and for all $i$ there exists $\theta_i<\theta$ such that $S_i\in K_{\theta_i}$. Let $(O_i)_{i\in I}$ be a family of disjoint open sets such that for all $i$ we have $S_i\subseteq O_i$, $O = \cup_{i\in I}O_i$ and $C$ be the complement of $O$ (notice $C$ is then a closed set). We have $S = C\cup \Cup_{i\in I}(O_i\cap \overline{S_i})$. By Lemma \ref{LambdaClosedIntersectionOpen} we know that for each $i$ we have $O_i\cap \overline{S_i}\in \Lambda_{\theta_i}$, which means that $\Cup_{i\in I}(O_i\cap \overline{S_i}$ belongs to $\Lambda_\theta$ which concludes the proof.
\end{proof}

\subsection{Equivalence between the Hausdorff difference hierarchy and the fine Hausdorff hierarchy}

This section is devoted to proving the Hausdorff difference hierarchy and fine Hausdorff hierarchy actually coincide. We have two preliminary lemmas: one states that the collections $\mathcal{D}_\theta$ are closed under open union, and the other states that the collections $\mathcal{D}_\theta$ contain exactly the sets that can be written as open unions of sets lower in the hierarchy, similarly to how $\Lambda_\theta$ is defined.

\begin{lemma}\label{DthetaClosedOpenUnion}
	For all ordinals $\theta$, $\mathcal{D}_\theta$ is closed under open union.
\end{lemma}

\begin{proof}
	Let $\theta$ be an ordinal and let $(S_i)_{i\in I}$ be a family of sets in $\mathcal{D}_\theta$. There exists a family of open sets $(O_i)_{i\in I}$ such that for each $i$ $S_i\subseteq O_i$ and for each $i$ there exists a family of open sets $(P^i_\eta)_{\eta<\theta}$ such that $S_i = D_\theta((P^i_\eta)_{\eta<\theta})$. Let $S = \Cup_{i\in I} S_i$. For all $\eta<\theta$, let $P_\eta= \Cup_{i\in I} O_i\cap P^i_\eta$ (notice $(P_\eta)_{\eta<\theta}$ is an increasing family of open sets). We want to show that $S = D_\theta((P_\eta)_{\eta<\theta})$.
	
	First let us prove that $S \subseteq D_\theta((P_\eta)_{\eta<\theta})$. Let $\rho\in S$. There exists $i\in I$ such that $S\in S_i$. We know $\rho\in O_i$, which means that for all $\eta$ such that $\rho\in P^i_\eta$ we have $\rho\in O_i\cap P^i_\eta$. Moreover, we know that for all $j\neq i$ we have $\rho\notin O_j$. As a consequence, we have $\min\{\eta|\rho\in\Cup_{j\in I} O_j\cap P^j_\eta\} = \min\{\eta|\rho\in P^i_\eta\}$, and hence $\rho\in D_\theta((P_\eta)_{\eta<\theta})$.
	
	Now we want to show that $D_\theta((P_\eta)_{\eta<\theta})\subseteq S$. Let $\rho\in  D_\theta((P_\eta)_{\eta<\theta})$. Obviously there exists a unique $i\in I$ such that $\rho\in O_i$ (there exists at least one for all $j\neq i$ $O_j$ is disjoint from $O_i$). We want to show that $\rho\in S_i$. We have $\min\{\eta|\rho\in P_\eta\} = \min\{\eta|\rho\in\Cup_{j\in I} O_j\cap P^j_\eta\} = \min\{\eta|\rho\in O_i\cap P^i_\eta\} = \min\{\eta|\rho\in P^i_\eta\}$. Since $\rho\in S$, the smallest $\eta$ such that $\rho\in P_\eta$ is of parity opposite to that of $\theta$, and hence it is also the case for the smallest $\eta$ such that $\rho\in P^i_\eta$, which means that $\rho\in S_i$ and thus $\rho\in S$.
\end{proof}

\begin{lemma}\label{DthetaLimitUnion}
	For all limit ordinals $\theta$, a set $S$ belongs to $\mathcal{D}_\theta$ if and only if it is the open union of sets in $\cup_{\eta<\theta}\mathcal{D}_\theta$.
\end{lemma}

\begin{proof}
	Consider a limit ordinal $\theta$ and a family of sets $(S_i)_{i\in I}$ such that for all $i\in I$ there exists $\theta_i<\theta$ such that $S_i\in\mathcal{D}_{\theta_i}$, and such that there exists a family of disjoint open sets $(O_i)_{i\in I}$ such that for all $i\in I$ we have $S_i \subseteq O_i$. Since by Lemma \ref{HausdorffHierarchyIncreasing} we have for all $i$ $\mathcal{D}_{\theta_i}\subseteq \mathcal{D}_{\theta_i + 1}$ we can suppose without loss of generailty that for all $i\in I$ $\theta_i$ is of the same parity as $\theta$ (that is, even). Let $S = \Cup_{i\in I} S_i$. We want to prove that $S\in \mathcal{D}_\theta$. For all $i\in I$, there exists an increasing family of open sets $(O^i_\eta)_{\eta<\theta_i}$ such that $S_i = D_{\theta_i}((O^i_\eta)_{\eta<\theta_i})$. For all $\eta<\theta$, let $O_\eta = \cup_{i\in I, \eta' \leq \eta} O_i \cap O^i_{\eta'}$. We will show that $S = D_\theta((O_\eta)_{\eta<\theta})$. 
	
	Indeed, consider $\rho\in S$ and let us prove $\rho\in D_{\theta}((O_\eta)_{\eta<\theta})$. There exists $i\in I$ such that $\rho\in S_i$, which means that $\rho\in O_i$ and hence for all $j\neq i$ we have $\rho\notin O_j$. This means that for all $\eta<\theta$, $\rho\in O_\eta$ if and only if $\eta<\theta_i$ and $\rho\in O^i_\eta$, and thus the smallest $\eta$ such that $\rho \in O_\eta$ is also the smallest $\eta$ such that $\rho\in O^i_\eta$, which ultimately means that $\rho\in D_{\theta}((O_\eta)_{\eta<\theta})$ since both $\theta$ and $\theta_i$ are even.
	
	Let us now consider $\rho\in D_\theta((O_\eta)_{\eta<\theta})$ and let us prove there exists some $i\in I$ such that $\rho\in S_i$. There exists $\eta<\theta$ such that $\rho\in O_\eta$, and hence there exists $i\in I$ such that $\rho\in O_i$ and $\rho\in O^i_\eta$. Since for all $j\neq i$ the sets $O_i$ and $O_j$ are disjoint we know that for all $\eta<\theta$ $\rho\in O_\eta$ if and only if $\eta<\theta_i$ and $\rho\in O^i_\eta$. This means that the smallest $\eta$ such that such that $\rho \in O^i_\eta$ is also the smallest $\eta$ such that $\rho\in O_\eta$, and thus $\rho\in S_i$ since both $\theta$ and $\theta_i$ are even.
	
	Conversely, let us consider some set $S\in \mathcal{D}_\theta$. There exists an increasing family of open sets $(O_\eta)_{\eta<\theta}$ such that $S = D_\theta((O_\eta)_{\eta<\theta})$. Let $\mathcal{L} = \{h\in\mathcal{H}\mid \exists \eta<\theta, \cyl(h) \subseteq O_\eta\}$ be the set of histories whose cylinder is included in some $O_\eta$ for $\eta<\theta$. Let $\mathcal{F}$ be the set of histories which belong to $\mathcal{L}$ and are minimal with regards to the prefix relation. For $f\in\mathcal{F}$, let $\theta_f = \min\{\eta<\theta\mid\cyl(f)\subseteq O_\eta\}$. Finally, for $f\in \mathcal{F}$, if $\theta_f$ is even then let $S_f = D_{\theta_f + 2}((\cyl(f)\cap O_\eta)_{\eta<\theta_f +2})$, and if $\theta_f$ is odd then let $S_f = D_{\theta_f + 1}((\cyl(f)\cap O_\eta)_{\eta<\theta_f+1})$. Notice this means that for all $f\in\mathcal{F}$, a play $\rho$ belongs to $S_f$ if and only if the smallest $\eta$ such that $\rho\in\cyl(f)\cap O_\eta$ is smaller than $\theta_f$ and is odd. We shall prove that $S = \Cup_{f\in \mathcal{F}} S_f$ (notice we obviously have that for all $f\in\mathcal{F}$ $S_f\subseteq \cyl(f)$, which ensures this is indeed an open union as all histories in $\mathcal{F}$ are incomparable with each other with regards to the prefix relation).
	
	Consider $\rho\in S$ and let us show there exists $f\in\mathcal{F}$ such that $\rho\in S_f$. There exists $\eta < \theta$ such that $\rho\in O_\eta$, and hence there exists $f\in\mathcal{F}$ such that $f$ is a prefix of $\rho$. This means in particular that $\rho\in\cyl(f)$, and as a consequence we have that the smallest $\eta$ such that $\rho\in O_\eta$ is smaller than $\theta_f$, and is also the smallest $\eta$ such that $\rho\in\cyl(f)\cap O_\eta$. Since $\theta$ is even, this also means that $\rho\in S_f$.
	
	Now consider $f\in\mathcal{F}$, $\rho\in S_f$ and let us show $\rho\in S$. Obviously $\rho\in \cyl(f)$, and the smallest $\eta$ such that $\rho\in O_\eta$ is odd, which means that $\rho\in S$.
\end{proof}

Finally, we can prove our main result about the two hierarchies, justifying that the fine Hausdorff hierarchy is actually a refinement of the Hausdorff difference hierarchy:

\begin{theorem}[Theorem \ref{THM:fineHausdorffIsHausdorffCore}]
	For all ordinals $\theta$, we have $\mathcal{D}_\theta=\Lambda_\theta$.
\end{theorem}

\begin{proof}
	We prove the result by induction on ordinals $\theta$. First, we obviously have $\mathcal{D}_1 = \Lambda_1$ as both are the collections of all the open sets.
	
	Now consider some ordinal $\theta$ such that for all $\eta<\theta$ we have $\mathcal{D}_\eta = \Lambda_\eta$. We shall show prove that $\mathcal{D}_\theta = \Lambda_\theta$.
	
	If $\theta$ is a limit ordinal, then this is a straightforward consequence of Lemma \ref{DthetaLimitUnion} and the induction hypothesis.
	
	Now suppose that $\theta$ is a successor ordinal.
	\begin{itemize}
		\item Consider $S\in\Lambda_\theta$. By Corollary \ref{LambdaOpenUnionOfKSmaller}, we know that there exists a family of sets $(S_i)_{i\in I}$ such that for all $i\in I$ we have $ S_i\in K_{\theta-1}$ and $S = \Cup_{i\in I} S_i$. By definition for all $i\in I$ we have $\overline{S_i}\in\Lambda_{\theta-1}$ which means that $\overline{S_i}\in\mathcal{D}_{\theta-1}$. By Lemma \ref{HausdorffComplementIncreases} we then have that for all $i\in I$, $S_i\in\mathcal{D}_\theta$, and finally by Lemma \ref{DthetaClosedOpenUnion} we have $S\in\mathcal{D}_\theta$.
		\item Consider now $S\in\mathcal{D}_\theta$. By Lemma \ref{DthetaSminusO} there exists an open set $O$ and a set $S'$ in $\mathcal{D}_{\theta-1}$ such that $S = O\setminus S' = O\cap \overline{S'}$. By the induction hypothesis we have $S'\in\Lambda_{\theta-1}$, which implies that $\overline{S'}\in K_{\theta-1}$ and thus $\overline{S'}\in\Lambda_\theta$. Finally, by Corollary \ref{LambdaClosedIntersectionOpen} we can conclude that $S\in\Lambda_\theta$.
	\end{itemize}
\end{proof}

\subsection{Equivalence of representations of $\Delta_2^0$} \label{AP:equivalenceOfRepresentations}

\begin{theorem}[Theorem \ref{THM:delta02Representations}]
	Given a subset $S$ of $C^\omega$, the following propositions are equivalent:
	\begin{itemize}
		\item[$(1)$] $S\in\Delta_2^0$;
		\item[$(2)$] $S$ belongs to the Hausdorff difference hierarchy;
		\item[$(3)$] $S$ belongs to the fine Hausdorff hierarchy;
		\item[$(4)$] there exists an eventually constant labelling funcion $l$ such that $S = 1_l$.
	\end{itemize}
\end{theorem}

\begin{proof}[Proof of Theorem \ref{THM:delta02Representations}]
	\begin{itemize}
		\item $(1)\Leftrightarrow (2)$ is the hausdorff-Kuratowski theorem, a proof of which can be found in \cite{Kechris95}.
		\item $(2)\Leftrightarrow (3)$ is a straightforward consequence of Theorem \ref{THM:fineHausdorffIsHausdorffCore}.
		\item To prove $(3)\Rightarrow (4)$, we want to show that the collections of sets $1_l$ where $l$ is an eventually constant labelling function contains the open sets and is closed under both open union and complementation. First notice that any open set $O$ with generating family $\mathcal{F}$ is equal to $1_{l_\mathcal{F}}$ where $l_{\mathcal{F}}(w) = 1$ if and only if $w$ has a prefix in $\mathcal{F}$. Then, if $l$ is an eventually constant labelling function, the function $\overline{l}$ defined by $\overline{l}(h) = 1-l(h)$ is also an eventually constant labelling function and we have $\overline{1_l} = 1_{\overline{l}}$. Finally, consider a countable family of eventually labeling function $(l_i)_{i\in I}$ and a family of disjoint open sets $(O_i)_{i\in I}$ associated with generating families of words $(\mathcal{F}_i)_{i\in I}$. Consider then the labelling function $l$ defined by:
		\begin{itemize}
			\item if there exists $i\in I$ and $f\in \mathcal{F}_i$ such that $f$ is a prefix of $w$ then $l(w) = l_i(h)$;
			\item else $l(w) = 0$
		\end{itemize}
		Notice that since the $O_i$'s are disjoint from one another we cannot have $f_1\in \mathcal{F}_i$ and $f_2$ in $\mathcal{F}_j$ such that $i\neq j$ and $f_1$ is a prefix of $f_2$, which ensures $l$ is well-defined. We claim that we then have $1_l = \Cup_{i\in I} 1_{l_i}$.
		\item Finally, let us show $(4)\Rightarrow (1)$. Consider an eventually constant labelling function $l$ and the set $W = 1_l$. For $n\in\mathbb{N}$, let the set $O_n$ be the open set generated by the family of histories $\{h \mid l(h) = 1 \text{ and } |h|\geq n\}$. We have $W = \cap_{n\in\mathbb{N}}O_n$, and thus $W\in\Pi^0_2$. Conversely, for all $n\in\mathbb{N}$, let $C_n$ be the closed set defined as $\rho\in C_n$ if and only if all prefixes of $\rho$ of length equal to or greater than $n$ have label $1$. We also have $W = \cup_{n\in\mathbb{N}}C_n$, and hence $W\in\Sigma^0_2$.
	\end{itemize}
\end{proof}

\subsection{Equivalence between Büchi/co-Büchi conditions and $\Pi_2^0$/$\Sigma_2^0$} \label{app:Buchi}

We provide here the proof of Lemma \ref{LM:BuchiIsPi02}, which states that the sets that can be described via a Büchi condition are exactly the sets that belong to $\Pi_2^0$.

\begin{proof}[Proof of Lemma \ref{LM:BuchiIsPi02}]
	Suppose that $S$ belongs to $\Pi_0^2$: there exists a sequence $(O_n)_{n\in\mathbb{N}}$ of open sets such that $S=\cap_{n\in\mathbb{N}}O_n$. Without loss of generality, we can suppose that this sequence is increasing and $O_0=C^\omega$. Then define the coloring function $c$ in the following way:
	\begin{itemize}
		\item if $u\in C^*$ is such that $\cyl(u)\subseteq S$ then $c(u)=1$
		\item else if $u$ is such that $\max\{n\in\mathbb{N}|\cyl(u)\subseteq O_n\} > \max\{n\in\mathbb{N}|\cyl(u_{<|u|-1})\subseteq O_n\}$ then $c(u)=1$
		\item else $c(u)=0$.
	\end{itemize}
	Now take a word $w$ that belongs to $S$: either there exists an index $n$ such that $\cyl(u_{\leq n})\subseteq S$ or there exists an infinite number of indices $n$ such that $ \max\{n\in\mathbb{N}|\cyl(w_{\leq n})\subseteq O_n\} > \max\{n\in\mathbb{N}|\cyl(w_{<n})\subseteq O_n\}$ (because the set $\{n|\rho\in O_n\}$ is infinite). Conversely, if a word $w$ is such that there exists an infinite number of indices $n$ such that $c(u_{\leq n})=0$ then the set $\{n|u\in O_n\}$ is infinite and thus $u\in S$.
	
	Now suppose that there exists a coloring function $c$ such that $u\in S\Leftrightarrow \{n\in\mathbb{N}\mid c(\rho_{\leq n})=1\}$ is infinite. We define a sequence of open sets $(O_n)_{n\in\mathbb{N}}$ by $O_n = \{u\mid\exists k\geq n \text{ such that } c(u_{\leq k})=1\}$. One can then easily check that we have $S=\cap_{n\in\mathbb{N}}O_n$.
\end{proof}

\section{Proof of the main result}

\subsection{Proof for sets in $K_2$} \label{AP:caseK2}

We provide here the proof of Lemma \ref{LM:depthBounded}, which states that when the winning set $W$ is the union of a closed set and an open set ($W = C\cup O$), for all $\hpref\in\PrefC$ the value of $\depth(\hpref, h)$ for $h$ eligible is bounded.

\begin{proof}[Proof of Lemma \ref{LM:depthBounded}]
	Consider $\hpref$ in $\PrefC$ and $h_1, h_2$ in $\Gamma$ such that $h_1, h_2 \notin \PrefC$ and $W_{\hpref} \subseteq W_{h_1} \subseteq W_{h_2}$. By construction for all $l$ we have $h_1l\in TC(\hpref, h_1) \Leftrightarrow h_2l \in TC(\hpref, h_2)$. Moreover, for all $l$ we have $W_{h_1l} \subseteq W_{h_2l}$ by Lemma \ref{orderTransitive}. This ensures that for all $l$ belonging to $\mathcal{F}_{\hpref, h_1}$, since we have $W_{h_1l} = (A\times B)^\omega$ then we have $W_{h_2l} = (A\times B)^\omega$ too, and since $h_2l\notin \PrefC$ this means that $h_2l$ has a prefix in $\mathcal{F}$ and thus a prefix in $\mathcal{F}_{\hpref, h_2}$. This shows $\depth(\hpref, h_2) \leq \depth(\hpref, h_1)$.
	
	Consider then a finite set $\mathcal{S}$ of histories in $\Gamma\setminus\PrefC$ such that for all $h\in\mathcal{S}$ we have $W_{\hpref}\subseteq W_{h}$ and for all $h'\in\Gamma\setminus\PrefC$ such that $W_{\hpref}\subseteq W_h$ there exists $h\in \mathcal{S}$ such that $W_h\subseteq W_{h'}$. As mentioned, we then have that $\depth(\hpref,h')\leq\depth(\hpref,h)$, which ensures that for all $h'$ eligible we have $\depth(\hpref, h')\leq \max\{\depth(\hpref,h)\mid h\in\mathcal{S}\}$.
	
\end{proof}

\subsection{Proof for sets in $\Lambda_\theta$} \label{AP:inductiveCaseLambda}

We first show that before our strategy emulates some $s_l$, the memory state is exatcly the current history:

\begin{lemma}\label{LM:muSameAsReal}
	For all histories $h$ compatible with $s_f$, for all $t\in T_s'$, we have $\mu(h) = t \Leftrightarrow h = t$.
\end{lemma}

\begin{proof}
	A straightforward induction on $h$ allows us to prove this result.
\end{proof}

This guarantees that when we reach a memory state $l\in\mathcal{l}$, the current history is indeed $l$, which means that emulating $s_l$ from this point guarantees the win. This is exactly what our strategy does, as stated by the following Lemma:

\begin{lemma}\label{LM:sfEmulatessl}
	For all $l$ in $\mathcal{L}$ and all $h$ such that $lh$ is compatible with $s_f$ we have $s_f(lh) = s_l(h)$.
\end{lemma}

\begin{proof}
	The result is obtained by a straightforward induction on $h$.
\end{proof}

\begin{corollary}\label{COR:sfWinsIfLPrefix}
	For all $\rho\in (A\times B)^\omega$ compatible with $s_f$, if there exists $l\in \mathcal{L}$ such that $l$ is a prefix of $\rho$ then $\rho\in W$.
\end{corollary}

\begin{proof}
	This is a direct application of Lemma \ref{LM:sfEmulatessl} and the fact that for all $l\in\mathcal{L}$ $s_l$ wins from $l$.
\end{proof}

Finally, we can prove our strategy is winning:

\begin{lemma}
	The finite-memory strategy $s_f$ is a winning strategy from $\varepsilon$.
\end{lemma}

\begin{proof}
	By Corollary \ref{COR:sfWinsIfLPrefix}, we only need to prove that for all plays $\rho$ compatible with $s_f$, $\rho$ has a prefix in $\mathcal{L}$. Consider then some play $\rho$ compatible with $s_f$. By Lemma \ref{LM:muSameAsReal} and since $T_s'$ is finite, we know there exists some $k\in\mathbb{N}$ such that for all $n>k$ we have $\mu(\rho_{\leq n})\notin T_s'$. Consider the smallest such $k$. We know that we have $\mu(\rho_{\leq k} = \rho_{\leq k} \in T_s$. Then by the definition of $\mu$, either $\mu(\rho_{\leq k+1}) = \rho_{\leq k+1}$ or $\mu(\rho_{\leq k+1}) = m_l$ where $\rho_{\leq k+1} = l \in \mathcal{L}$. We know the former is impossible as $\mu(\rho_{\leq k+1})$ does not belong to $T_s'$ by the definition of $k$ and hence does not belong to $T_s$ as$M\cap T_s = T_s'$ (because all $M_l$'s are disjoint from $T_s$). This means that the latter is true, and thus $\rho_{\leq k + 1}\in \mathcal{L}$, which concludes the proof.
\end{proof}

\subsection{Proof for sets in $K_\theta$}\label{AP:inductiveCaseK}

We provide here a proof of the induction step for Theorem \ref{thebigtheorem} for sets in $K_\theta$ for some ordinal $\theta\geq 1$. Suppose that $W\in K_\theta$ and Theorem \ref{thebigtheorem} is true for all winning sets belonging to $\Lambda_\eta$ or $K_\eta$ for all $\eta<\theta$, as well as for all winning sets belonging to $\Lambda_\theta$. We suppose that Player $1$ has a winning strategy from $\varepsilon$ and we want to show that she also has a finite-memory winning strategy.

By Lemma \ref{LM:KUnionCLambda}, there exists a closed set $C$ and a set $L$ in $\Lambda_\theta$ such that $W = C\cup L$. By Theorem \ref{THM:fineHausdorffIsHausdorffCore}, $L$ belongs to $\mathcal{D}_\theta$ and hence there exists an increasing family of open sets $(O_\eta)_{\eta<\theta}$ such that $L = D_\theta((O_\eta)_{\eta<\theta})$. For $h\in\mathcal{H}$, if there exists $\eta$ such that $\cyl(h)\subseteq O_\eta$ then we call the smallest such $\eta$ the rank of $h$. Else we say that $h$ is of rank $\theta$. We write $\rank(h)$ for the rank of $h$. Notice that if $h$ is a prefix of $h'$ then the rank of $h$ is larger than the rank of $h'$, and hence along a play $\rho$ there exists a unique ordinal $\eta$ that is the rank of infinitely many prefixes of $\rho$: $\rho$ then belongs to $W$ if and only if $\eta$ is of parity opposite to that of $\theta$.

Similarly to the case for $K_2$, we split $\mathcal{H}$ into two sets: one contains the histories that have at least one continuation in $C$ (this is the set $\PrefC$) and the other contains the histories that do not. For all histories $h$ in $\Gamma\setminus\PrefC$, we have $W_h = h^{-1}L\in\Lambda_\theta$, which means by Lemma \ref{LM:closedIntersectionCylinder} that $W_h\in\Lambda_\theta$. This means that, by the induction hypothesis, there exists a finite-memory strategy $s_h$ that wins from $h$.

Let $\hpref$ in $\Gamma\cap\PrefC$ and $h$ in $\Gamma\setminus\PrefC$ such that $W_{\hpref}\subseteq W_h$ and the rank of $h$ is of the same parity as $\theta$. As before we call $TC(\hpref, h)$ the set $\{ hl \mid \hpref l \in \PrefC\}$ consisting of the finite continuations from $\hpref$ that belong to $\PrefC$, but rooted in $h$. Notice that all infinite continuations $\rho$ such that for all $i$ we have $\hpref\rho_{\leq i}\in\PrefC$ are such that $\hpref\rho\in C$, and hence $\hpref\rho\in W$ which means that $h\rho\in W$, and finally $h\rho\in L$ since $h\notin\PrefC$ (which means that all winning continuations of $h$ belong to $L$ because they cannot belong to $C$). Since $h$'s rank is of the same parity as $\theta$, we thus know that $TC(\hpref, h)$ contains a family of histories $\mathcal{F}_{\hpref, h}$ of rank strictly smaller than $h$ and such that all infinite branches of $TC(\hpref, h)$ have a finite prefix in $\mathcal{F}_{\hpref, h}$. By Kőnig's lemma we can assume that this family is finite. We call $\depth(\hpref, h)$ the maximal length of the finite continuations $l$ such that $hl \in\mathcal{F}_{\hpref,h}$.

In contrast with the case for $K_2$ however, we cannot prove that for all $\hpref\in\PrefC$ the value of $\depth(\hpref,h)$ for $h$ satisfying the hyptheses above is bounded. We then need to introduce more complex ideas to build our finite-memory winning strategy.

For $\hpref$ in $\PrefC$ and $\eta\leq\theta$, we call $\Slice_\eta(\hpref)$ the set of histories $h$ of rank $\eta$ such that $W_{\hpref}\subseteq W_h$. We let $\Min_\eta(\hpref)$ be a finite subset of $\Slice_\eta(\hpref)$ such that for all $h'$ in $\Slice_\eta(h)$ there exists $h\in\Min_\eta(\hpref)$ such that $W_h\subseteq W_{h'}$ (we know such a finite set exists because $\subseteq$ induces a well partial order on induced winning sets). Let $\mathcal{O}$ be the set of ordinals $\eta$ such that $\eta\leq\theta$ and the parity of $\eta$ is the same as that of $\theta$. For $\hpref\in\PrefC$, we build an increasing family of sets of histories $(N_\eta(\hpref))_{\eta\in\mathcal{O}}$ in the following inductive way:
\begin{itemize}
	\item $N_{\eta_0}(\hpref) = \Min_{\eta_0}(\hpref)$ where $\eta_0$ is the smallest ordinal in $\mathcal{O}$,
	\item for $\eta>\eta_0$ then we let $N_{<\eta}(\hpref) = \cup_{\zeta<\eta,\zeta\in\mathcal{O}} N_\zeta(\hpref)$ and $N_\eta(\hpref) = N_{<\eta}(\hpref)\cup\{h\in\Min_\eta \mid \forall h'\in N_{<\eta}(\hpref), W_{h'}\nsubseteq W_h\}$. 
\end{itemize}
Informally, we construct the family $N_\eta(\hpref)$ by adding the histories in $\Min_\eta(\hpref)$ whose winning set is incomparable with that of the histories we have already considered. This process eventually converges to a finite set, as stated by the following Lemma:

\begin{lemma}
	For $\hpref\in\PrefC$, the set $N_\theta(\hpref)$ is finite.
\end{lemma}

\begin{proof}
	Suppose that it is infinite. As for all $\eta$ the set $\Min_\eta(\hpref)$ is finite, this means that there exists a sequence $(h_n)_{n\in\mathbb{N}}$ of histories in $N_\theta(\hpref)$ such that $(\rank(h_n))_{n\in\mathbb{N}}$ is strictly increasing. By the construction of $(N_\eta(\hpref))_{\eta\in\mathcal{O}}$, this means that for all $i,j$ such that $i<j$ we have $W_{h_i}\nsubseteq W_{h_j}$, which is impossible since $\subseteq$ induces a well partial order on the induced winning sets.
\end{proof}

This allows us to define $\depth(\hpref) = \max\{\depth(\hpref, h) \mid h\in N_\theta(\hpref)\}$ for $\hpref\in\PrefC$. The construction of the winning finite-memory strategy then follows in a similar fashion to the case for $K_2$. As we will see later, the set $N_\theta (\hpref)$ satisfies three essential properties which constitute the groundwork for our proof:
\begin{itemize}
	\item it is finite, which allows us to define the value of $\depth(\hpref)$;
	\item it provides under-approximations for all the winning sets induced by histories $h$ such that the rank of $h$ is of the same parity as $\theta$ and $W_{\hpref} \subseteq W_h$;
	\item its construction allows us to prove inductively that, for all such histories $h$, if we play the strategy we define from $h$ starting in memory state $\hpref$, then we can guarantee the win.
\end{itemize}

Consider a finite family of finite-memory strategies $s_i=(M_i,\sigma_i,\mu_i,m_i)$, indexed by a set $I$ and such that for all histories $h$ in $\Gamma\setminus\PrefC$ there exists some $i\in I$ such that $s_i$ wins from $h$. Such a family exists because by the partial order hypothesis there exists a finite set $\mathcal{S}$ of histories in $\Gamma\setminus\PrefC$ whose induced winning sets provide under)approximation for all winning sets induced by histories in $\Gamma\setminus\PrefC$. It then only remains to consider one finite-memory strategy for each history in $\mathcal{S}$. 

Consider also a finite family $(\hpref_j)_{j\in J}$ of histories in $\Gamma\cap\PrefC$ indexed by $J\subseteq\mathbb{N}$ such that for all histories $\hpref$ in $\PrefC\cap\Gamma$ there exists $j\in J$ such that $W_{\hpref_j}\subseteq W_{\hpref}$. For all $j\in J$, let $T_j = \{\hpref_j l \in\PrefC \mid |l|\leq\depth(\hpref_j)\}$. Up to renaming, we can suppose that the $T_j$'s are dsijoint from one another. We build our finite-memory winning strategy $s = (M,\sigma,\mu,m_0)$ in the following way:
\begin{itemize}
	\item $M = \cup_{i\in I} M_i \cup \cup_{j\in J} T_j$;
	\item for $m\in M_i$ we let $\sigma(m) = \sigma_i (m)$;
	\item for $t\in T_j$ we let $\sigma(t)$ be any non-losing action from $t$;
	\item for $m\in M_i$ and $(a,b)\in A\times B$ we let $\mu(m, (a,b)) = \mu_i(m,(a,b))$;
	\item for $t\in T_j$ and $(a,b)\in A\times B$ such that $a = \sigma(t)$:
	\begin{itemize}
		\item if $t(a,b)\in T_j$ then $\mu(t,(a,b)) = t(a,b)$;
		\item else if $t(a,b)\in\Gamma\setminus\PrefC$ then there exists $i\in I$ such that $s_i$ wins from $t(a,b)$: we let $\mu(t,(a,b))=m_i$;
		\item else if $t(a,b)\in\Gamma\cap\PrefC$ then there exists $j\in J$ such that $W_{\hpref_j}\subseteq W_{t(a,b)}$ and we let $\mu(h,(a,b))=\hpref_j$;
	\end{itemize}
	\item $m_0 = \hpref_j$ where $j\in J$ is such that $W_{\hpref_j}\subseteq W_\varepsilon$.
\end{itemize}

Once again, our memory function is built in such a way that if the memory state belongs to $T_j$ for some $j\in J$ then it provides an under-approximation for the winning set induced by the current history. This time however we have to be more precise: if $t\in \cup_{j\in J}T_j$ is such that $W_t \subseteq W_h$, then for all finite continuations $h'$ such that $\mu(t,h')$ (the memory state obtained by following our strategy from memory state $t$ against history $h'$) belongs to $\cup_{j\in J}T_j$, $\mu(t,h')$ provides an under-approximation for the induced winning set of $hh'$. 

\begin{lemma}\label{LM:memContainedRealGeneralCase}
	For all $t\in \cup_{j\in J} T_j$ and $h\in\mathcal{H}$, if $W_t\subseteq W_h$ then for all $h'$ such that $\mu(t,h')\in \cup_{j\in J} T_j$ we have $W_{\mu(t,h')}\subseteq W_{hh'}$.
\end{lemma}

\begin{proof}
	The proof is by induction on $h'$. First we have $\mu(t,\varepsilon) = t$ and by hypothesis $W_t\subseteq W_h$. Consider now $h'\in\mathcal{H}$ such that $\mu(t,h')\in T_j$ for some $j\in J$ and such that $W_{\mu(t,h')}\subseteq W_{hh'}$. Let $(a,b)\in A\times B$ such that $\mu(t, h'(a,b))\in T_{j'}$ for some $j'\in J$.	Then,
	\begin{itemize}
		\item if $\mu(t, h'(a,b)) = \mu(t, h')(a,b)$ then the desired result follows by Lemma \ref{orderTransitive},
		\item else we must have $\mu(t, h')(a,b)\in\Gamma\cap\PrefC$ (else we would not have $\mu(t, h'(a,b))\in T_{j'}$) and $\mu(t, h'(a,b))=\hpref_{j'}$ with $j'$ such that $W_{\hpref_{j'}}\subseteq W_{\mu(t,h')(a,b)}$, and $W_{\mu(t,h')(a,b)}\subseteq W_{hh'(a,b)}$ once again by Lemma \ref{orderTransitive}, which ensures the result.
	\end{itemize}
\end{proof}

\begin{corollary}\label{COR:memContainedRealGeneralCase}
	If $\mu(h)\in T_j$ for some $j\in J$ then we have $W_{\mu(h)}\subseteq W_h$.
\end{corollary}

\begin{proof}
	This is a direct application of Lemma \ref{LM:memContainedRealGeneralCase}, as by construction we have $W_{m_0}\subseteq W_\varepsilon$.
\end{proof}

Moreover, if we ever reach a memory state in $M_i$ for some $i\in I$ then from this point onward our finite-memory strategy will emulate the strategy $s_i$:

\begin{lemma}\label{LM:sEmulatessi}
	For all $m\in M$, if $h\in\mathcal{H}$ is such that $\mu(m,h)\in M_i$ for some $i$ then for all $h'$ we have $\mu(m,hh') = \mu_i(\mu(m,h),h')$.
\end{lemma}

\begin{proof}
	This is shown by a straightforward induction on $h'$.
\end{proof}

\begin{corollary}\label{COR:ifsiWinsthensWins}
	For all $\in I$, if $h\in\mathcal{H}$ is such that $s_i$ wins from $h$ then $(M,\mu,\sigma,m_i)$ wins from $h$.
\end{corollary}

\begin{proof}
	By application of Lemma \ref{LM:sEmulatessi} and by the definition of $\sigma$, all plays $\rho$ compatible with $(M,\mu\sigma,m_i)$ are also compatible with $s_i$.
\end{proof}

By combining Lemma \ref{LM:memContainedRealGeneralCase} with Lemma \ref{LM:sEmulatessi}, we obtain that for any history $h\in\Gamma$, if $t\in T_j$ is such that the winning set induced by $t$ constitutes an under-approximation of the winning set induced by $h$, if by playing according to our finite-memory strategy from memory state $t$ in $h$ we reach a memory state that belongs to $M_i$ for some $i\in I$ then we can guarantee that if we keep playing according to $s$ then we will win:

\begin{lemma}\label{LM:memhmiThensiWins}
	For all $t\in \cup_{j\in J}T_j$ and $h\in\mathcal{H}$, if $h$ is such that $W_t\subseteq W_h$ then for all $h'$ compatible with $(M,\sigma,\mu,t)$ such that $\mu(t,h') = m_i$ for some $i\in I$ then $s_i$ is winning from $hh'$.
\end{lemma}

\begin{proof}
	Consider $h'_0$, the smallest prefix of $h'$ such that $\mu(t,h'_0)\in M_i$. By Lemma \ref{LM:sEmulatessi} we know that $h'_0$ is also the smallest prefix of $h'$ such that $\mu(h'_0)\notin\cup_{j\in J}T_j$ (because if $\mu(t,h')\in M_i$ then no prefix $h''$ of $h'$ can be such that $\mu(t,h'')\in M_{i'}$ with $i'\neq i$). We will show that $s_i$ wins from $hh'_0$. Notice that if $s_i$ wins from $hh'_0$ then $s_i$ wins from $hh'$: indeed, let $l$ be the history such that $h' = h'_0l$. As we have $\mu(t, hh'_0) = \mu(t, hh') = m_i$, by Lemma \ref{LM:sEmulatessi}, this means that $\mu_i(m_i,l) = m_i$. If $s_i$ wins from $hh'_0$ then for all $l'$ compatible with $s_i$, the strategy $(M_i,\sigma_i,\mu_i,\mu_i(l'))$ must win from $hh'_0l'$, and in particular $s_i$ must win from $hh'_0l = hh'$.
	
	Let then $h'_0 = l'_0(a,b)$ with $l'\in\mathcal{H}$ and $(a,b)\in A\times B$ (we cannot have $h'_0 = \varepsilon$ as $\mu(t,\varepsilon) = t$ by definition). Obviously $\mu(t,l'_0)\in \cup_{j\in J} T_j$. Since $\mu(t,l'_0(a,b))\in M_i$, by construction we have $\mu(t,l'_0) = m_i$ and $s_i$ wins from $\mu(t,l'_0)(a,b)$. Furthermore, by Lemma \ref{LM:memContainedRealGeneralCase} we have $W_{\mu(t,l'_0)}\subseteq W_{hl'_0}$, and hence $W_{\mu(t,l'_0)(a,b)}\subseteq W_{hl'_0(a,b)} = W_{hh'_0}$ by Lemma \ref{orderTransitive}, which means that $s_i$ wins from $hh'_0$ as well.
\end{proof}

\begin{corollary}\label{COR:memhmiThensWins}
	For all $t\in \cup_{j\in J}T_j$ and $h\in\mathcal{H}$, if $h$ is such that $W_t\subseteq W_h$ then for all $h'$ compatible with $(M,\sigma,\mu,t)$ such that $\mu(t,h') = m_i$ for some $i\in I$ then $(M,\sigma,\mu,m_i)$ is winning from $hh'$.
\end{corollary}

\begin{proof}
	This is a straightforward application of Corollary \ref{COR:ifsiWinsthensWins} and Lemma \ref{LM:memhmiThensiWins}.
\end{proof}

\begin{corollary}\label{COR:memhmiThensiWins}
	If $h\in\mathcal{H}$ is compatible with $s$ and is such that $\mu(h) = m_i$ for some $i\in I$ then $s_i$ is winning from $h$.
\end{corollary}

\begin{proof}
	This is a direct application of Lemma \ref{LM:memhmiThensiWins}, as by construction we have $W_{m_0}\subseteq W_\varepsilon$.
\end{proof}

\begin{corollary}
	If $h\in\mathcal{H}$ is compatible with $s$ and is such that $\mu(h)\in M_i$ for some $i\in I$ then $(M,\sigma,\mu,\mu(h))$ is winning from $h$.
\end{corollary}

\begin{proof}
	If $\mu(h)\in M_i$ then by construction $h$ has a prefix $h_0$ such that $\mu(h_0) = m_i$. By Lemma \ref{COR:memhmiThensiWins} we know that $s_i$ wins from $h_0$, and hence by Lemma \ref{LM:sEmulatessi} we have that $(M,\sigma,\mu,\mu(h))$ wins from $h$.
\end{proof}

Finally, we can show that if by playing according to $s$ we reach a history that is outside of $\PrefC$ then we can keep playing according to $s$ and win:

\begin{lemma}\label{LM:sWinsFromOutsidePrefC}
	For all $h\in\mathcal{H}\setminus \PrefC$, if $h$ is compatible with $s$ and $\mu(h) = \hpref_j$ for some $j\in J$ then $(M,\sigma,\mu,\hpref_j)$ wins from $h$.
\end{lemma}

\begin{proof}
	For $j\in J$ and $\eta\leq\theta$ we let $\mathcal{P}(\hpref_j,\eta)$ be the following property: for all $h\in\Slice_\eta(\hpref_j)$, $(M,\sigma,\mu,\hpref_j)$ wins from $h$. We shall prove by induction on $\eta$ that for all $\eta\leq\theta$, for all $j\in J$ we have $\mathcal{P}(\hpref_j,\eta)$.
	
	For the initial case ($\eta = 0$), let $j\in J$.
	\begin{itemize}
		\item If $\theta$ is even, then all histories $h$ such that $\rank(h) = 0$ are such that $W_h = \emptyset$. Since $\hpref_j\in\Gamma$, it follows that $\Slice_0(\hpref_j) = \emptyset$ and hence we have $\mathcal{P}(\hpref_j,0)$.
		\item If $\theta$ is odd, then all histories $h$ such that $\rank(h) = 0$ are such that $W_h = (A\times B)^\omega$, which means that any strategy wins from $h$ and hence we have $\mathcal{P}(\hpref_j,0)$.
	\end{itemize}
	
	Consider now an ordinal $\eta$ such that $\eta>0$ and $\eta\leq\theta$ and suppose that for all $j\in J$ and for all $\eta'<\eta$ we have $\mathcal{P}(\hpref_j,\eta')$. 
	
	Suppose that $\eta$ is of parity opposite to that of $\theta$. Let $j\in J$,  $h\in\Slice_\eta(\hpref_j)$ and $\rho$ be a play compatible with $(M,\sigma,\mu,\hpref_j)$. We want to prove that $h\rho\in W$. If there exists $k\geq 0$ such that $\mu(\hpref_j,\rho_\leq k) = m_i$ for some $i\in I$ then $h\rho$ belongs to $W$ by Corollary \ref{COR:memhmiThensWins}. Else from the construction of $\mu$ we have that for all $k\geq 0$, $\mu(\hpref,\rho_{\leq k})\in\cup_{j'\in J}T_{j'}$. If for all $k$ we have $\rank(h\rho_{\leq k}) = \eta$ then $h\rho\in W$ and we have our result. Suppose then that there exists $k$ such that $\rank(h\rho_{\leq k}) = \eta' < \eta$. By construction there exists some integer $n \geq k$ such that $\mu(\hpref_j, \rho_{\leq n}) = \hpref_{j'}$ for some $j'\in J$. Then by Lemma \ref{LM:memContainedRealGeneralCase}, we have that $W_{\hpref_{j'}}\subseteq W_{h\rho_{\leq n}}$, which means that $h\rho_{\leq n}\in\Slice_{\eta'}(\hpref_j')$. By $\mathcal{P}(h_{j'},\eta')$, we can conclude that $(M,\sigma,\mu,h_{j'})$ wins from $h\rho_{\leq n}$ and thus $h\rho = h\rho_{\leq n}\rho_{>n}\in W$ (because $\rho$ is compatible with $(M,\sigma,\mu,h_j)$ and $\mu(h_j,\rho{\leq n})=h_{h_j'}$, which means that $\rho_{>n}$ is compatible with $(M,\sigma,\mu,h_{j'})$).
	
	Suppose now that $\eta$ and $\theta$ have the same parity. Let $j\in J$ and let $h\in\Slice_{\eta}(\hpref_j)$. There exists a history $h_0$ in $\Min_{\eta}(\hpref_j)$ such that $W_{h_0}\subseteq W_h$. We shall prove that $(M,\sigma,\mu,\hpref_j)$ wins from $h_0$, which implies that $(M,\sigma,\mu,\hpref_j)$ wins from $h$. If $h_0\notin N_\eta(\hpref_j)$ then by construction there exists $\eta'<\eta$ and $h_1\in N_{\eta'}(\hpref_j)$ such that $W_{h_1}\subseteq W_{h_0}$. By $\mathcal{P}(\hpref_j,\eta')$ we know that $(M,\sigma,\mu,\hpref_j)$ wins from $h_1$ and hence wins from $h_0$. Suppose then that $h_0\in N_\eta(\hpref_j)$ and let $\rho$ be a play compatible with $(M,\sigma,\mu,\hpref_j)$. We want to show that $h_0\rho$ belongs to $W$. As before, if there exists $k\geq 0$ such that $\mu(\hpref_j,\rho_{\leq k}) = m_i$ for some $i\in I$ then $h_0\rho$ belongs to $W$ by Corollary \ref{COR:memhmiThensWins}. Else from the definition of $M$ as the disjoint union of $\cup_{i\in I} M_i$ and $\cup_{j'\in J} T_{j'}$ we have that for all $k\geq 0$, $\mu(\hpref,\rho_\leq k)\in\cup_{j'\in J}T_{j'}$. In particular notice that for every $k$ such that $k\leq \depth(\hpref_j)$ we have $\mu(\hpref_j,\rho_\leq k) = \hpref_j\rho_{\leq k}\in T_j$. Let $d = \depth(\hpref_j)$. Since $h_0\in N_\eta(\hpref_j)$ and $N_\eta(\hpref_j)\subseteq N_\theta(\hpref_j)$ we have $d\geq \depth(\hpref_j,h_0)$. Recall that by definition this means that for any $l$ such that $|l|\geq d$ and $\hpref_jl\in \PrefC$ we have $\rank(h_0l)<\eta$. Since $T_j\subseteq\PrefC$ by definition, we thus have that $\rank(h_0\rho_{\leq d})  < \eta$. Again by construction, there exists $j'\in J$ such that $\mu(\hpref_j,\rho_{\leq d+1}) = \hpref_{j'}$. By Lemma \ref{LM:memContainedRealGeneralCase} we have that $W_{\hpref_{j'}}\subseteq W_{h_0\rho_{\leq d+1}}$, and by the properties of the rank we have that $\rank(h_0\rho_{\leq d +1}) = \eta' < \eta$ (because $\rank(h_0\rho_{\leq d +1}) \leq \rank(h_0\rho_{\leq d})$). Finally, by $\mathcal{P}({h_{j'},\eta'})$ we know that $(M,\sigma,\mu,\hpref_{j'})$ wins from $h_0\rho_{\leq d +1}$, which means that $h_0\rho\in W$ (because $\rho$ is compatible with $(M,\sigma,\mu,\hpref_j)$ and $\mu(\hpref_j,\rho_{\leq d+1}) = \hpref_{j'}$, which implies that $\rho_{> d +1}$ is compatible with $(M,\sigma,\mu,\hpref_{j'})$) and allows us to conclude.
\end{proof}

We finally have all the preliminary results we need to conclude this induction step:

\begin{theorem}
	The finite-memory strategy $s$ wins from $\varepsilon$.
\end{theorem}

\begin{proof}
	Let $\rho$ be a play compatible with $s$. We want to show that $\rho\in W$. If for all $n\in\mathbb{N}$ we have $\rho_{\leq n}\in \PrefC$ then $\rho\in C$ and hence $\rho\in W$. Else if there exists some $n\in\mathbb{N}$ such that $\mu(\rho)\in M_i$ for some $i\in I$ then by Corollary \ref{COR:memhmiThensWins} we have $\rho\in W$. Else by construction there exists some $n\in\mathbb{N}$ such that $\rho_{\leq n}\notin\PrefC$ and such that there exists $j\in J$ such that $\mu(\rho_{\leq n}) = \hpref_j$. By Corollary \ref{COR:memContainedRealGeneralCase} we know that $W_{\hpref_j}\subseteq W_{\rho_{\leq n}}$ and we can conclude by Lemma \ref{LM:sWinsFromOutsidePrefC} that $\rho\in W$.
\end{proof}

\section{Tightness of the main result}\label{AP:tightness}

\subsection{Büchi counter-example}

Here we provide an example of a game whose winning set satisfies the well partial order property and belongs to $\Pi_2^0$. While Player $1$ does have a winning strategy, she does not have any fnite-memory one.

\begin{example}
	Let $A$ be a finite set of at least two elements and $w$ a disjunctive sequence on $A\times\{0\}$. We define the labeling function $l: (A\times\{0\})^*\to \{0,1\}$ in the following way: $l(h) = 1$ if and only if there exists $h_0$ and $h'$ such that $h = h_0h'$ and $h'$ is the longest factor of $h$ that is also a prefix of $w$. Let then $W$ be the set defined by $\rho\in W$ if and only if infinitely many prefixes $h$ of $\rho$ are such that $l(h) = 1$ and consider the game $(A,\{0\},W)$.
	
	Let us study the elements of $W$. We want to show that a play $\rho$ belongs to $W$ if and only if its set of factors contains infinitely many prefixes of $w$. A play $\rho$ belongs to $W$ if and only if it is such that infinitely many prefixes $h$ of $\rho$ are such that $l(h) = 1$. By definition if $l(h) = 1$ then the longest factor of $h$ that is also a prefix of $w$ is also a suffix of $h$. This means that if $l(h) = 1$ and $h'$ is such that $l(hh') = 1$ then $hh'$ has at least one factor which is also a prefix of $w$ and longer than any factor of $h$ satisfying the same property. In particular this means that the set of factors of $hh'$ contains strictly more prefixes of $w$ than the set of factors of $h$. Consequentially, if a play $\rho$ is such that it has infinitely many prefixes $h$ such that $l(h)=1$, then the set of factors of $\rho$ contains infinitely many prefixes of $w$. Conversely, let $\rho$ be such that the set of factors of $\rho$ contains infinitely many prefixes of $w$ and let $\rho = h_0h_1...$ where for each $n$ the word $h_0...h_n$ is the smallest prefix of $\rho$ such that it has a suffix which is a prefix of $w$ that is not a factor of $h_0...h_{n-1}$. Notice that if a prefix $p$ of $w$ is a factor of $h_0...h_n$ but not of $h_0...h_{n-1}$ then it is also longer than all prefixes of $w$ that are factors of $h_0...h_{n-1}$. This means that for all $n$ we have $l(h_0...h_n) = 1$ and hence $\rho\in W$.
	
	With this characterization of $W$, we will now see that this winning condition is prefix-independent: that is, for all histories $h$ we have $W_h = W$. Indeed, consider two histories $h$ and $h'$ and a play $\rho$ such that $h\rho\in W$. Let us split the set of factors of $h\rho$ into two sets: $F_>$ is the set of factors of $\rho$ while $F_<$ is the set of factors of $h\rho$ that overlap with $h$. Either $F_>$ or $F_<$ contains infinitely many prefixes of $w$, and we will show $F_>$ always does. Suppose that $F_<$ contains infinitely many prefixes of $w$. In particular, since $h$ is a finite word, there exists a suffix $h_0$ of $h$ such that infinitely many of these elements can be written as $h_0r$ with $r$ a prefix of $\rho$. This means that infinitely many prefixes $r$ of $\rho$ are such that $h_0r$ is a prefix of $w$, and hence we actually have $h_0\rho = w$. This means that $\rho$ is a suffix of $w$, and as such is also a disjunctive sequence, which ensures that infinitely many of its factors are prefixes of $w$. This means that $F_>$ always contains infinitely many prefixes of $w$, and hence $h'\rho\in W$. 
	
	Since the winning condition is prefix-independent, all the induced winning sets are equal: for all histories $h$ we have $W_h = W$. Hence it is obvious that $\subseteq$ induces a well partial on the induced winning sets. Moreover, all winning plays $\rho$ have infinitely many factors that are also prefixes of $w$, and thus are disjunctive sequences too and are such can only be irregular sequences. This means that, while Player $1$ obviously have a winning strategy from $\varepsilon$, she cannot have a finite-memory one (whose resulting play could only be a regular sequence).
\end{example}

\subsection{Co-Büchi counter-example}

Here we provide an example of a game whose winning set satisfies the well partial order property and belongs to $\Sigma_2^0$. While Player $1$ does have a winning strategy, she does not have any fnite-memory one.

\begin{example}
	Let $A$ be a finite set of at least two elements and $w$ an irregular sequence in $(A\times\{0\})^*$. Let $W = \{\rho\in (A\times \{0\})^\omega\mid\exists h\in (A\times \{0\})^*, \rho = h\rho_0 \text{ where } \rho_0 \text{ is a suffix of } w\}$. $W$ is the set of sequences which have a suffix in common with $w$. Consider then the game $(A,\{0\}, W)$.
	
	We shall now show that $W\in\Sigma^0_2$. We have $W = \cup_{h\in A^*}\cup_{\rho_0 \text{ suffix of } w} \{h\rho_0\}$. Since $A^*$ is countable and $w$ has a countable number of suffixes, $W$ is then a countable union of closed set, which means that $W\in\Sigma^0_2$.
	
	Obviously the winning condition of this game is prefix-independent, that is we have for all histories $h$ $W_h = W$, and hence $\subseteq$ trivially induces a well partial on the induced winning sets. Moreover as $W\neq \emptyset$ we know that Player $1$ has a winning strategy from $\varepsilon$. However, all winning play have a suffix which coincide with a suffix of $w$. As $w$ is irregular, all winning plays are irregular too, which means that there exists no finite-memory winning strategy for Player $1$ in this game.
\end{example}

\subsection{Virtual bar counter-example}

In this section we tackle the well partial order hypothesis and study one possible weakening of it: instead of $\subseteq$ inducing a well partial order on the histories of $\Gamma$, we only make the assumption that there exists a virtual bar for the induced winning set: there exists a finite family $S$ of sets such that every induced winning set includes  set in $S$. We provide an example with a closed winning set that shows this hypothesis does not lead to the same result, as Player $1$ has a winning strategy but no finite one.

\begin{example}
	Let $A = \{0,1,x,y\}$. For all $n\in\mathbb{N}$, let $v_n = (1,0)(0,0)^n$ and $u_n = (0,0)(1,0)^n$. Let $c_x = (x,0)v_0v_1...$ and $c_y = (y,0 u_0u_1...$. Consider an irregular sequence $s_n$ over $\{x,y\}$ and let $W$ be the set of words in $(A\times\{0\})^\omega$ such that $\rho\in W$ if and only if either there exists a sequence $(w_n)_{n\in\mathbb{N}}$ of non-empty finite words in $(A\times\{0\})^*$ such that for each $n$ we have $w_n\sqsubset c_{s_n}$ and $\rho = w_0w_1w_2...$, or there exists $k\in\mathbb{N}$ and a finite sequence $(w_n)_{n<k}$ of non-empty finite words in $(A\times\{0\})^*$ such that for each $n<k$ we have $w_n\sqsubset c_{s_n}$ and $\rho = w_0...w_{k-1}c_{s_k}$. As in every one-player games where the winning set is non-empty, Player $1$ has a winning strategy from $\varepsilon$ in $(A,\{0\}, W)$.
	
	For a finite word $h$ in $(A\times\{0\})^*$ beginning with either $(x,0)$ or $(y,0)$, we call \textit{$xy$-decomposition} of $h$ the sequence $(w_n)_{n<k}$ of non-empty finite words such that each $w_n$ begins with either $(x,0)$ or $(y,0)$ and the letters $(x,0)$ or $(y,0)$ do not appear later in $w_n$, and $h = w_0...w_k$.
	
	Let us show $W$ is a closed set. We will show that $\rho$ does not belong to $W$ if and only if there exists a finite prefix $h$ of $\rho$ such there exists no finite sequence $(w_n)_{n<k}$ of non-empty finite words in $(A\times B)^*$ such that for each $n$ $w_n\sqsubset c_{s_n}$ and $h = w_0...w_{k-1}$. The complement of $W$ is then the open set generated by these such histories. Notice the desired property is equivalent to the fact that either $h$ does not begin with either $(x,0)$ or $(y,0)$ or the $xy$-decomposition of $h$ $(w_n)_{n<k}$ does not satisfy $w_n \sqsubset c_{s_n}$ for all $n<k$.
	
	Let $\rho\in (A\times\{0\})^\omega$. Let $h$ be a finite prefix of $\rho$ such that there exists no finite sequence $(w_n)_{n<k}$ of non-empty finite words in $(A\times\{0\})^*$ such that for each $n$ $w_n\sqsubset c_{s_n}$ and $h = w_0...w_{k-1}$. Obviously there cannot be any sequence $(w_n)_{n\in\mathbb{N}}$ of non-empty finite words in $(A\times\{0\})^*$ such that for each $n$ $w_n\sqsubset c_{s_n}$ and $\rho = w_0w_1w_2...$, nor any $k\in\mathbb{N}$ and finite sequence $(w_n)_{n<k}$ of non-empty finite words in $(A\times\{0\})^*$ such that for each $n<k$ we have $w_n\sqsubset c_{s_n}$ and $\rho = w_0...w_{k-1}c_{s_k}$, as we could use such a sequence to build a suitable sequence for $h$. This means that $\rho\notin W$.
	
	Let then $\rho\notin W$. Obviously if $\rho$ does not start with either $(x,0)$ or $(y,0)$ then $\rho$ satisfies the desired property (with its prefix of length $1$). Else, 
	\begin{itemize}
		\item if the letters $\{(x,0),(y,0)\}$ appear infinitely often along $\rho$, then let $(w_n)_{n\in\mathbb{N}}$ be the sequence of non-empty words such that $\rho = w_0w_1...$ and for each $n$ $w_n$ begins with either $(x,0)$ and $(y,0)$ and the letters $\{(x,0),(y,0)\}$ do not appear in $w_n$ except for the first letter. Since $\rho\notin W$, there exists $k$ such that $w_k$ is not a prefix of $c_{s_k}$. This means in particular that $w_0...w_k$ is a finite prefix of $\rho$ such that $(w_n)_{n<k+1}$ is its $xy$-decomposition and it does not satisfy $w_n\sqsubset c_{s_n}$ for all $n<k+1$.
		\item else the letters $\{(x,0),(y,0)\}$ appear finitely often along $\rho$. Let then $(w_n)_{n<k}$ be the sequence of non-empty words such that $\rho = w_0w_1...w_{k-1}$ and for each $n$ $w_n$ begins with either $(x,0)$ or $(y,0)$ and the letters $\{(x,0),(y,0)\}$ do not appear in $w_n$ except for the first letter. Since $\rho\notin W$, either there exists $n<k-1$ such that $w_n$ is not a prefix of $c_{s_n}$, or $w_{k-1}\neq c_{s_{k-1}}$. In the first case, we use arguments similar to the previous item. In the second case, there exists a prefix $w'$ of $w_{k-1}$ that is not a prefix of $c_{s_n}$, and we use similar arguments on $w_0...w_{k-2}w'$.
	\end{itemize}
	
	Now that we know $W$ is a closed set, let us show that Player $1$ does not have a finite-memory winning strategy from $\varepsilon$ in $(A,\{0\},W)$. Indeed, consider such a finite-memory winning strategy $(M,\sigma,\mu,m_0)$ and let $\rho$ be its outcome.
	\begin{itemize}
		\item If there exists a sequence $(w_n)_{n\in\mathbb{N}}$ of non-empty finite words in $(A\times\{0\})^*$ such that for each $n$ $w_n\sqsubset c_{s_n}$ and $\rho = w_0w_1w_2...$, then Player $1$ also has a finite-memory winning strategy from $\varepsilon$ in $(A,\{0\},\{s_0s_1...\})$, which is impossible since $(s_n)_{n\in\mathbb{N}}$ is irregular.
		\item Else there exists $k\in\mathbb{N}$ and a finite sequence $(w_n)_{n<k}$ of non-empty finite words in $(A\times\{0\})^*$ such that for each $n<k$ we have $w_n\sqsubset c_{s_n}$ and $\rho = w_0...w_{k-1}c_{s_k}$. This means that Player $1$ has a finite-memory winning strategy from $\varepsilon$ in $(A,\{0\},\{c_{s_k}\})$, which is also impossible. 
	\end{itemize}
	
	$W$ is hence a closed set with virtual bar $\{c_x,c_y\}$ which is such that there exists no finite-memory winning strategy for Player $1$ in $(A,\{0\},W)$.
\end{example}

\subsection{What about the opponent?}

In this section, we want to answer the following question: given a two-player game $(A,B,W)$ where $W$ belongs to the Hausdorff difference hierarchy and such that $\subseteq$ induces a well partial order on the induced winning sets for Player $1$, if Player $2$ has a winning strategy for $\varepsilon$ then is it the case that he also has a finite-memory winning strategy? We provide here a counter-example to show this is not the case as soon as $W\in\Lambda_2$.

\begin{example}
	Let $A = B = \{0,1\}$. Let $W = \cup_{n\in\mathbb{N}}\cup_{k\leq n} (\{0\}\times B)^n(\{1\}\times B) [(A\times \{0\})^k(A\times \{1\})(A\times B)^\omega + (A\times \{0\})^\omega]$. Informally, Player $1$ wins if the following happens:
	\begin{itemize}
		\item first Player $1$ plays $0$ a number $n$ of times, and then plays $1$;
		\item then Player $2$ either plays $0$ a number $k\leq n$ of times, and then plays $1$, or Player $2$ plays $0$ for ever.
	\end{itemize}
	Hence, to win Player $2$ must wait until Player $1$ plays action $1$, then play $0$ for strictly longer than he has waited, and only after play $1$.
	We can thus divide the histories of this game into four groups:
	\begin{itemize}
		\item the first group is composed of the histories along which Player $1$ has only played $0$. The induced winning set of these histories is wholly determined by their length, and for any two histories belonging to this group, the longest history is always associated with a bigger winning set. Hence $\subseteq$ induces a well partial order over the induced winning sets associated with the histories of this group.
		\item the second group is composed of the histories along which Player $1$ has played $0$ a number $n$ of times, then played $1$, and since then Player $2$ has played $0$ a number $k$ of times and has not played $1$. The induced winning set for the histories of this group is wholly determined by the value $n-k$: if $n-k\leq 0$ then the induced winning set associated with the history is $(A\times \{0\})^\omega$, while if $n-k>0$ then the induced winning set associated with the history is $(A\times \{0\})^\omega\cup\cup_{i<n-k} (A\times\{0\})^i(A\times \{1\})$. In particular, if two histories $h$ and $h'$ associated with parameters $(n_h,k_h)$ and $(n_{h'}, k_{h'})$ in this group are such that $0\leq n_h - k_h \leq n_{h'} - k_{h'}$ then we have $W_h \subseteq W_{h'}$. This means that $\subseteq$ induces a well partial order over the induced winning sets associated with the histories of this group.
		\item the third group consists of the histories along which Player $1$ has played $0$ a number $n$ of times, then played $1$, and since then Player $2$ has played $0$ a number $k\leq n$ of times and then played $1$. The induced winning set for the histories of this group is always equal to $(A\times B)^\omega$, and hence $\subseteq$ trivially induces a well partial order over said induced winning sets.
		\item finally, the fourth group is the group of the histories along which Player $1$ has played $0$ a number $n$ of times, then played $1$, and since then Player $2$ has played $0$ a number $k> n$ of times and then played $1$. The induced winning set for the histories of this group is always equal to $\emptyset$, which means that $\subseteq$ trivially induces a well partial over said induced winning sets.
	\end{itemize}
	We have seen that $\subseteq$ induces a well partial order over the induced winning sets associated with the histories of each group. This means that it also induces a well partial order over the induced winning sets of the game as a whole.
	
	One can also easily check that Player $2$ has a winning strategy from $\varepsilon$, which consists in waiting until Player $1$ plays action $1$ and then play action $0$ for long enough until he can finally win by playing $1$. However, as the number of times he has to play action $0$ increases as Player $1$ waits before playing action $1$, he does not have a finite-memory winning strategy.
\end{example}

\end{document}